\documentclass[a4paper,11pt]{article}

\usepackage{amsthm,amsmath,amssymb}
\usepackage{mathtools}

\def \hring {{\hat{\cal R}}}
\def \ring {{\cal R}}

\def \ord {{\rm ord}}
\def \Ord {{\rm Ord}}

\def \hu {\hat{u}}
\def \hv {\hat{v}}
\def\Id{\mbox{Id}}
\def \lieg {{\mathfrak{g}}}

\def \raph {{\, \xmapsto{\varphi}\, }}

\newcommand\cF{{\mathcal F}}

\newcommand\cR{{\mathcal R}}

\newcommand\cI{{\mathcal I}}

\newcommand\cT{{\mathcal T}}

\newcommand\cM{{\mathcal M}}

\newcommand\Mon{{\mathcal Mon(\cR)}}
\newcommand\DerR{{\mathcal Der(\cR)}}
\newcommand\DerF{{\mathcal Der(\cF)}}
\newcommand*{\defeq}{\mathrel{\vcenter{\baselineskip0.5ex \lineskiplimit0pt
                     \hbox{\scriptsize.}\hbox{\scriptsize.}}}%
                     =}

\newcommand\cFS{{\cF[\cS,\cS^{-1}]}}
                     
\newtheorem{Def}{Definition}
\newtheorem{Ex}{Example}
\newtheorem{Thm}{Theorem}
\newtheorem{Pro}{Proposition}

\newtheorem{Rem}{Remark}

\def \k {\C}

\def \bu {{\bf u}}
\newcommand\Sg{{\Sigma}}
\newcommand\cS{{\mathcal S}}
\newcommand{\field}[1]{\mathbb{#1}}

\newcommand{\C}{\field{C}}
\newcommand{\Z}{\field{Z}}
\newcommand{\N}{\field{N}}

\newcommand{\D}{\field{D}}
\newcommand\T{{\mathcal S}}
\newcommand\fF{{\mathfrak{F}}}
\newcommand\fA{{\mathfrak{A}}}

\date{}
\begin{document}
\title{Perturbative Symmetry Approach for Differential-Difference Equations.}

\author{Alexander V. Mikhailov$^{+*}$,
Vladimir S. Novikov$^{\dagger}$ and Jing Ping Wang$ ^\& $\\
$+$School of Mathematics, University of Leeds, UK\\
$*$P. G. Demidov
Yaroslavl State University, RF\\
$\dagger$ School of Mathematics, Loughborough University, UK\\
$\&$ School of Mathematics, Statistics and Actuarial Science, University of Kent, UK
}

\maketitle
\begin{abstract}

We propose a new method for solution of the integrability problem for evolutionary differential-difference equations of arbitrary order. It enables us to produce necessary integrability conditions, to determine whether a given equation is integrable or not, and to advance in classification of integrable equations. In this paper we define and develop symbolic representation for the difference polynomial ring, difference operators and formal series. In order to formulate necessary integrability conditions, we introduce a novel quasi-local extension of the difference ring. It enables us to progress in classification of integrable differential-difference evolutionary equations of arbitrary order. In particular, we solve the problem of classification of integrable equations of order $(-3,3)$ for the important subclass of quasi-linear equations and produce a list of 17 equations satisfying the necessary integrability conditions.   For every equation from the list we present an infinite family of integrable higher order relatives. Some of the equations obtained
%from the list and higher order hierarchies 
are new.
\end{abstract}

\section{Introduction}

The problem to determine whether a given equation is integrable (testing for integrability) and much more difficult problem to give an exhaustive description of all integrable cases for a certain type of equations up to invertible transformations (the classification problem) are central in the theory of integrable systems. In this paper we study evolutionary differential-difference equations
\begin{equation}\label{eq1int}
 u_t=F(u_p,\ldots,u_q)
\end{equation}
for a function $u=u(n,t)$ of one discrete variable $n\in\Z$ and a continuous independent variable $t\in\C$. Here
we use the standard notations
\begin{eqnarray*}
 u_t=\partial_t(u), \quad u_k=\cS^k u(n,t)=u(n+k,t)
\end{eqnarray*}
and $\cS$ is the shift operator. Existence of an infinite algebra of (infinitesimal) symmetries of equation (\ref{eq1int}) we take as the definition of its integrability. There are many alternative views on integrability, including  existence of multi-soliton solutions, ``regular dynamics'', Painlev\'e property of finite dimensional reductions, integrable continuous limits, etc., which are useful but difficult to formalise or inconclusive. The symmetry approach provides us with a rigorous framework enabling to formulate necessary integrability conditions  which are suitable for solution of the classification problem for  equations of arbitrary order. 

In the case of partial differential equations (PDEs) the symmetry approach proved to be successful for  classification of evolutionary equations and system of equations \cite{mr86i:58070}-\cite{mr93b:58070}. Its further development, the perturbative symmetry approach for PDEs, based on the natural degree grading structure and  symbolic representation of the differential ring \cite{mr58:22746}-\cite{mr2001h:37147} enabled us to extend the method to nonlocal and/or non-evolutionary equations such as
the Benjamin-Ono equation and the Camassa-Holm equation  \cite{mn1, mn2, mnw07}.
Symbolic representation was successfully used for global classification of integrable scalar homogeneous evolutionary equations \cite{mr99g:35058}. We refer to, for example, the review paper \cite{MNW3} and the recent book \cite{sok-book} for detail discussion of classification problems for integrable PDEs and related publications. 

The classification of integrable differential-difference equations has not enjoyed the same success as for partial differential equations so far.   The first classification result was obtained by Yamilov in 1983 \cite{yam83} for  general differential-difference equations of order $(-1,1)$:
$$
u_t=f(u_{-1}, u, u_1), \qquad \frac{\partial f}{\partial u_1}\neq 0, \qquad \frac{\partial f}{\partial u_{-1}}\neq 0.
$$
Yamilov derived and used for classification necessary conditions for existence of higher order generalised symmetries and conservation laws.
%The integrability conditions were derived based on the existence of higher order generalised symmetries and conservation laws. 
Since then this method has been developed further and used for the classification of other important types of differential-difference equations, including Toda and relativistic Toda-type equations \cite{Yami0}.
%As a result, exhaustive lists of such types of integrable equations have been obtained. All these developments 
%The method and results were reviewed in \cite{Yami0}. 

Recently,  Garifullin, Yamilov and Levi gave a partial classification result for the five-point differential-difference equations
\cite{Yami1}. They produced a complete list of quasi-linear equations of order $(-2,2)$, i.e. equations of the form
\begin{equation}\label{yamil1}
u_t=A(u_{-1},u, u_1) u_2+B(u_{-1},u, u_1) u_{-2} +C(u_{-1},u, u_1),
\end{equation}
which admit a symmetry of order $(-4,4)$. Their impressive classification list contains 31 equations up to autonomous point transformations, including some new equations, which are all proved to be integrable by the following-up study.  Their list contains equations of two types, namely equations which admit symmetries of orders $(-n,n)$ for all $n\in\N$, 
and equations which admit only even order of symmetries $(-2n, 2n), n\in \N$. Thus the resulting list depends on the assumption on the order of a symmetry that equation \eqref{yamil1} admits. This, list  would represent a complete classification of integrable equations, if it is shown that any integrable equation \eqref{yamil1} necessarily admits a symmetry of order $(-4,4)$. The latter is a challenging problem which cannot be tackled by the methods used in \cite{Yami1}.

To formulate integrability conditions  which are suitable for any lacunae in the sequence of symmetries,
%independent on algebra of symmetries structure
we   develop here a perturbative symmetry approach for differential-difference equations. The adaptation of methods previously used for PDEs is far not straightforward. It requires building a symbolic representation of difference rings and rings of difference operators, quasi-local extensions of rings and formal pseudo-difference series with quasi-local coefficients.  
For a differential-difference equation admitting an infinite algebra of symmetries Adler proved   existence of a formal recursion operator \cite{adler14}. We define a canonical formal recursion operator and for an integrable equation prove its existence and uniqueness. In symbolic representation its coefficients can be found explicitly for any equation, but if the equation admits an infinite algebra of symmetries then the coefficients must be quasi-local.   
Quasi-locality of the coefficients of the canonical formal recursion operator are  universal integrability conditions for differential-difference equations of arbitrary orders. We demonstrate the power of our method by solving the classification problem for an important family of quasi-linear differential-difference equations of order $(-3,3)$.

The paper starts with basic algebraic setting for the study of evolutionary differential-difference equations and algebras of their symmetries. In Section \ref{sec2} we discuss different grading for the difference polynomial ring and  its evolutionary derivations, difference   operators, and formal pseudo-difference series.  An introduction to the symmetry approach is given in Section \ref{sec3}. Approximate symmetries and approximate integrability are defined in Section \ref{secAprox}.  In Section \ref{sec4}, we define symbolic representation 
of a  difference polynomial ring, difference operators and formal series and formulate criteria of approximate integrability using symbolic representation.
For a given evolutionary differential-difference equation, either polynomial or represented by a formal series, we   give a recursive formula for the coefficients of its symmetries (Theorem \ref{theorsym}). The coefficients are uniquely determined by the linear part of the symmetry. This result can be used to test the existence of fixed order symmetries and to derive the necessary integrability conditions if the linear part of a symmetry is known.

For integrable PDEs the coefficients of a formal recursion operator must belong to differential field. These  necessary integrability conditions are  independent of possible lacunae in the hierarchy of symmetries. 
The proof is based on the existence of fractional powers of formal pseudo-differential series with coefficients in the differential field \cite{mr86i:58070}-\cite{mr93b:58070}.
In the case for difference operators or difference formal series, fractional powers   with coefficients in the difference ring (or field) may not exist \cite{MWX}. It motivates us to introduce a quasi-local extension of the difference ring (Section \ref{sec51}). Then the universal integrability conditions can be formulated as the conditions on the coefficients of the canonical formal recursion operator: the coefficients must be quasi-local (Section \ref{sec52}).

In Section \ref{sec6} we give a complete classification integrable  differential-difference equations of the form
\begin{equation}\label{gen1}
u_{t}=u_3f(u_2,u_1,u)-u_{-3}f(u_{-2},u_{-1},u)+g(u_2,u_1,u)-g(u_{-2},u_{-1},u),\qquad f(0,0,0)\ne 0,
\end{equation}
where $f,g$ are polynomial functions or formal series. We list only equations (\ref{gen1}) which do not admit symmetries of order $(-1,1)$ and $(-2,2)$. The latter are   known equations, since they are members of well known integrable hierarchies studied in 
\cite{Yami0, Yami1}.
Our list  consists on 17 equations satisfying necessary integrability conditions (Theorem \ref{class3}). We claim that the list is complete and have shown that all equations from the list are integrable. They either can be reduced to known integrable equations by difference substitutions (not invertible), or admit Lax representations. To the best of our knowledge, the list contains two genuinely new equations:
\begin{eqnarray}
u_{t}&=&(u^2+1)(u_3 \sqrt{u_1^2+1} \sqrt{u_2^2+1} - u_{-3} \sqrt{u_{-1}^2+1} \sqrt{u_{-2}^2+1}).\label{square}\\ 
u_{t}&=&u (u_2u_3 +u u_1-u u_{-1}-u_{-2}u_{-3})- u(u_2+u_1-u_{-1}-u_{-2}),\label{newadler}
\end{eqnarray}

For each of these $17$ equations, we found an infinite family of integrable equations of arbitrary high order. For instance, equation \eqref{square} is a member of the family of integrable equations
\begin{eqnarray}\label{eqrootn}
u_t=(1+u^2)(u_n\prod_{k=1}^{n-1}\sqrt{1+u_k^2}-u_{-n}\prod_{k=1}^{n-1}\sqrt{1+u_{-k}^2}), \quad  n\in \mathbb{N}.
\end{eqnarray}
Equation \eqref{square} is the $3$--relative (i.e. $n=3$) of this family. The $2$--relative ($n=2$) was discovered in classification of $(-2,2)$ order integrable equations \cite{Yami1} and $1$--relative is the well known modified Volterra equation.
For any two distinct values of  $k$ the corresponding $k$--relative non-linear equations from the same family are not compatible.
In  Section \ref{sec63} we have shown that equation (\ref{eqrootn}) has a Lax representation $L_t=[A,L]$ with 
\begin{eqnarray*}
L=Q^{-1}P,\quad A=L_{+}-L_{+}^{\dagger},\quad 
 Q=u- u_1 \sqrt{1+u^2}\ \cS^{-1}, \quad P=\left(u \sqrt{1+u_1^2}\ \cS-u_1\right) \cS^{n-1},
\end{eqnarray*}
where $L_{+}$ is the part with non-negative powers of $\cS$ in the Laurent formal difference series $L$ and $\dagger$ denotes the formal adjoint operator.

We conclude the paper with a short summary and discussion. In particular, we include Adler's Lax representation (with his kind permission \cite{Adler4}) for the integrable family of equations which includes equation \eqref{newadler}.

\section{Derivations, difference operators and formal series}\label{sec2}

The main objects of our study are evolutionary differential-difference equations and algebras of their symmetries. Although the phase space of such systems is infinite dimensional, each equation of the system relates a finite number of dynamical variables which we will treat as independent. In order to develop a rigorous theory we will use elements of the theory of difference rings, difference operators and formal series, including their symbolic representation. It will also enable us to introduce quasi-local extensions of the difference rings and formulate verifiable integrability conditions for differential-difference equations.

\subsection{Difference ring and its evolutionary derivations}

We define the polynomial ring $\cR=\k [\bu]$ and the corresponding field of fractions  $\cF=\k (\bu)$ of the infinite set of (commutative) variables $ \bu=\{u_n\,|\, n\in\Z \}$. We will often omit index zero at $u_0$. The set of all monomials $\Mon =\{u_{n_1}^{m_1} u_{n_2}^{m_2}\cdots u_{n_k}^{m_k}\,|\, n_i\in\Z,\ m_i\in\Z_{\geqslant 0}\}$ is the additive basis in $\cR$.

There is a natural automorphism $\cS$ of the field $\cF$, 
which we call the {\em shift operator}, defined as 
\[
 \cS: a( u _k,\ldots , u _r)\mapsto a( u _{k+1},\ldots , u _{r+1}),\quad 
\cS:\alpha\mapsto\alpha, \qquad a( u _k,\ldots , u _r)\in \cF,\ \ \alpha\in\k.
\]

The field $\cF$   equipped with   the automorphism $\cS$ is a 
difference field. 

The reflection 
$\cT$ of the lattice $\Z$ defined by
\[
 \cT: a( u _k,\ldots , u _r)\mapsto a(u_{-k},\ldots ,u_{-r}),\quad 
\cT:\alpha\mapsto\alpha, \qquad a( u _k,\ldots , u _r)\in \cF,\ \ \alpha\in\k,
\]
is another automorphism of $\cF$. The 
composition $\cS\cT\cS\cT=\Id$ is the identity map. Thus the automorphisms
$\cS,\cT$ generate the infinite dihedral group $\D_\infty$ and the infinite cyclic subgroup 
generated by $\cS$ is normal in $\D_\infty$.
The automorphism $\cT$ defines a $\Z_2$ grading of the difference field $\cF$:
\begin{equation}\label{grT}
  \cF =\cF _0\oplus\cF _1,\qquad \cF _0\cdot\cF _0=\cF _0,\ \ 
\cF _0\cdot\cF _1=\cF _1,\ \ \cF _1\cdot\cF _1=\cF _0, 
\end{equation}
where $ \cF _k=\{a\in\cF \,|\, \cT(a)=(-1)^k a\}$. The ring $\cR\subset\cF$ inherits the same grading  $\cR=\cR_0\oplus\cR_1$.

A derivation  $\partial $ in the field $\cF$ is a $\k$--linear map $\partial\,:\, \cF\mapsto\cF$   satisfying   Leibniz's law  
$$\partial (a\cdot b)=\partial(a)\cdot b+a\cdot\partial(b),\qquad  a,b\in\cF.$$ A set of all derivations in $\cF$ is denoted $\DerF$.
Let $\partial,\partial'\in\DerF$ be two derivations in $\cF$ then their commutator  $[\partial,\partial']\defeq\partial\circ\partial'-\partial'\circ\partial\in\DerF$ is also a derivation. Any three derivations $\partial,\partial',\partial''\in\DerF$ satisfy the Jacobi identity 
$$[\partial,[\partial',\partial'']]+[\partial',[\partial'',\partial]]+[\partial'',[\partial,\partial' ]]=0,$$
and therefore $\DerF$ is a Lie algebra over $\k$. A formal sum
\begin{equation}\label{derX}
 X= \sum_{n\in\Z}f^{(n)} \frac{\partial}{\partial u_{n}},\qquad f^{(n)}\in\cF 
\end{equation}
is a derivation in $\cF $. Its action on   $ X:\cF \mapsto\cF $ is well defined, since any element   $a\in \cF $ depends on a finite subset of variables, and thus the sum $X(a)$ contains only a finite number of non-vanishing terms. %The $\k$ linearity and the Leibniz rule are obviously satisfied.

Partial derivatives $\frac{\partial}{\partial u_i}\in\DerF $ are 
commuting derivations  satisfying the conditions
\begin{equation}\label{spart}
\cS \frac{\partial}{\partial u_{i}}= \frac{\partial}{\partial u_{i+1}} \cS, 
\qquad \cT \frac{\partial}{\partial u_{i}}= \frac{\partial}{\partial u_{-i}} 
\cT.
\end{equation}

A  derivation  $X\in\DerF$ is said to be {\em evolutionary} if it commutes with the shift operator $\cS$. For an evolutionary derivation it follows from the condition $X\circ\cS=\cS\circ X$ and (\ref{spart}) that all coefficients $f^{(n)}$ in (\ref{derX}) can be expressed $f^{(n)}=\cS^n (f )$ in terms of one element $f\in\cF $, which is called the {\em characteristic} of the evolutionary derivation. We will use notation
\begin{equation}
 \label{XF}
 X_f\defeq\sum_{i\in\Z}\cS^i(f) \frac{\partial}{\partial u_{i}}.
\end{equation}
for the evolutionary derivation corresponding to the characteristic  $f$. 

Evolutionary derivations form a Lie subalgebra of $\DerF$. Indeed, 
\[\begin{array}{l}
   \alpha X_f+\beta X_g=X_{\alpha f+\beta g},\qquad \alpha,\beta\in\k,\\
\phantom{}   [X_f,X_g]=X_{[f,g]},
  \end{array}
\]
where $[f,g]\in\cF $ denotes the Lie bracket
\begin{equation}\label{bracket}
 [f,g]=X_f(g)-X_g(f), 
\end{equation}
which is bi-linear, skew-symmetric and satisfying the Jacobi identity. Thus $\cF $ is a Lie algebra with Lie bracket defined by (\ref{bracket}). Evolutionary derivations with characteristics belonging to the ring $\cR$ form a  subalgebra of   $\DerR$.

The reflection $\cT$ acts naturally on evolutionary derivations
\[
 \cT:\, X_f\mapsto X_{\cT(f)}=\cT\cdot X_f\cdot \cT\, 
 \]

The polynomial ring $\cR$ is degree graded 
 \begin{equation}
 \label{gru}
 \cR=\bigoplus\limits_{n=0}^\infty\cR^{n}, 
\end{equation}
where $\cR^n$ is the set of all homogeneous polynomials of degree $n$.   The degree grading (\ref{gru}) and $\cT$--grading (\ref{grT}) are compatible and thus $\cR$ is a bi-graded ring
\begin{equation}\label{bigrR} \cR=\bigoplus\limits_{n=0}^\infty\cR^{n}_{0}\oplus\cR^{n}_{1}, \quad
 \cR^{n}_k=\cR^n\cap\cR_k\, .
 \end{equation}
 The homogeneous parts $\cR^n_k$ are additive groups and $\cR^n_k\cR^m_\ell=\cR^{n+m}_{(k+\ell)\mod2}$.
\iffalse 
 
Let us define Euler's  evolutionary derivation of $\cR$:
\begin{equation}
 \label{Xu}
 X_u=\sum_{k\in\Z}u_k\frac{\partial}{\partial u_k}
\end{equation}
 Monomials $a\in\Mon$ are eigenvectors of $X_u$
\[
 X_u(a)=\lambda  a,\quad a\in\Mon ,
\]
where the eigenvalue $\lambda=m_1+m_2+\cdots m_k$ is the total power of the monomial $a=u_{n_1}^{m_1} u_{n_2}^{m_2}\cdots u_{n_k}^{m_k}$. It enables us to define a grading of the ring $\cR$
\begin{equation}
 \label{gru}
 \cR=\bigoplus\limits_{n=0}^\infty\cR^{n}, \quad
 \cR^{n}=\{f\in\cR\, |\, X_u(f)=nf\}\, .
\end{equation}

 Moreover, since   $\cT\cdot X_u=X_u\cdot \cT$ the grading (\ref{grT}) is compatible with the above grading (\ref{gru}) and the ring $\cR$ has a bi-graded structure
\begin{equation}\label{bigrR} \cR=\bigoplus\limits_{n=0}^\infty\cR^{n}_{0}\oplus\cR^{n}_{1}, \quad
 \cR^{n}_k=\{f\in\cR\, |\, X_u(f)=nf,\ \cT(f)=(-1)^k f\}\, .
 \end{equation}
 \fi
It follows from the definition of the Lie bracket (\ref{bracket}) that the Lie algebra $\cR$ is also degree graded and bi-graded
\begin{equation}\label{bigrLie}
  [\cR^n ,\cR^m]\subset \cR^{n+m-1}\, ,
  \quad  [\cR^n_p,\cR^m_q]\subset \cR^{n+m-1}_{(p+q)\mod2}. 
\end{equation}

Every element $f$ of the polynomial ring $\cR$ can be uniquely represented as a sum of homogeneous components $f=\sum_{k\geqslant 0} f^{(k)},\ f^{(k)}\in \cR^k$ (some of the components may be equal to zero). The selection of the $k$--th  homogeneous component $f^{(k)}$ is a projection $\pi_k\,:\, \cR \mapsto\cR^k$ defined by
\begin{equation}
 \label{pi}
 \pi_k(f)=\pi_k \left(\sum_{i\geqslant 0} f^{(i)}\right)=f^{(k)}.
\end{equation}
Obviously $\pi_k\pi_s=\delta_{s,k}\pi_k$. Similarly one can define projectors to bi-graded homogeneous components.

\subsection{Difference, pseudo-difference operators and formal series}
Let  $a(u_n,\ldots,u_m)$ be a non-constant element of $\cF$, and we assume that $n\leqslant m$ are the minimal and the maximal index respectively in the sequence of its arguments. Then $X_f(a)$ can be represented by a finite sum
\[
 X_f(a)=\sum_{i=n}^m   \frac{\partial a}{\partial u_{i}}\cS^i(f)=a_*[f],
\]
where 
\begin{equation}\label{astar}
a_*\defeq\sum_{i=n}^m\frac{\partial a}{\partial u_i}\cS^i
\end{equation}
is the {\em  Fr\'echet derivative} of $a(u_n,\ldots,u_m)$ and  $a_*[f]$ is the Fr\'echet derivative of $a$ in the direction $f$.
Using the Fr\'echet derivative  we can represent the Lie bracket (\ref{bracket}) in the form:
\begin{equation}\label{LieF}
 [f\, ,\, g]=g_*[f]-f_*[g].
\end{equation}
It is obvious that
\[
(\cS a)_*=\cS \cdot a_*=\sum_{i=n}^m \cS\left(\frac{\partial a}{\partial u_i}\right) \cS^{i+1} \quad \mbox{and} \quad (\cT a)_*=\cT \cdot a_*=\sum_{i=n}^m\cT\left(\frac{\partial a}{\partial u_i}\right)\cS^{-i}.
\]

\begin{Def}\label{deford}  A difference operator $B$ of order ${\rm ord}\,  
B:=(n,m)$ with 
coefficients in
$\cF $ is a finite sum of the form
\begin{equation}\label{operB}
B= b^{(m)} \cS^{m}+b^{(m-1)} \cS^{m-1}+\cdots +b^{(n)} \cS^{n},\quad b^{(m)} b^{(n)}\ne 0,\  
b^{(k)}\in\cF , \ \ n\le 
m,
\ n,m\in\mathbb{Z}.
\end{equation}
The total order of $B$ is defined as ${\rm Ord}B=m-n$. The total order of 
the zero 
operator is minus infinity $\Ord\,\, 0:=-\infty$ by definition.
\end{Def}

The Fr\'echet derivative (\ref{astar}) is an example of a difference operator 
of order $(n,m)$ and total order ${\rm Ord}\, a_* =m-n$. For a non-constant element  
$f\in \cF $ the order and total 
order are defined as ${\rm ord}\, f_*$ and   ${\rm Ord}\, f_*$ respectively.

Difference operators form a unital ring $\cFS$ of 
Laurent 
polynomials in $\cS$ with coefficients in $\cF $, 
where multiplication is defined by 
\begin{equation} \label{smult}
 a\cS^n \cdot b\cS^m=a\cS^n(b)\cS^{n+m}=a b_n \cS^{n+m}.
\end{equation}
This multiplication is associative, but non-commutative. 

From the above definition it follows that if  $A$ is a difference operator of 
order ${\rm ord}\, A=(p,q)$, then ${\rm 
ord}\, (\cS^n\cdot  A\cdot \cS^m) =(p+n+m,q+n+m)$ and the total order ${\rm Ord}\, (\cS^n\cdot  
A\cdot \cS^m )={\rm Ord}\, 
A=q-p$. For any $A,B\in\cFS $ we have $\Ord\, (AB)=\Ord\, A+\Ord\, B$. 

For a difference operator $B$ given by \eqref{operB}, we define its adjoint operator $B^\dagger$ as
$$
B^\dagger=  \cS^{-m} \cdot b^{(m)} +\cS^{-m+1} \cdot b^{(m-1)}  +\cdots +\cS^{-n} \cdot b^{(n)}.
$$
Note that ${\rm ord}\ B^\dagger=(-m, -n)$ and ${\rm Ord}\ B^\dagger={\rm Ord}\ B$.

Below we define pseudo-difference (or rational) operators and skew-fields of formal series.
%, which we are going to use when we discuss formal recursion operators next section.  

\begin{Def}
 A rational (pseudo--difference) operator $M$ is defined as $M=AB^{-1}$ 
for some  $A,B\in\cFS$ and $B\ne 0$. The set of all 
rational operators is
 \[
  \fF=\{ AB^{-1}\,|\, A,B\in \cFS,\ B\ne0\}.
 \]
\end{Def}

\begin{Def} The sets  $\fF_L$  of Laurent and  $\fF_M$ of Maclaurin formal difference series  with coefficients in $\cF$   are 
%The skew fields $\fF_L$  and $\fF_M$ of the Laurent and Maclaurin formal difference series  with coefficients in $\cF$   are 
\[
 \fF_L=\left\{\sum_{n=k}^\infty a^{(-n)}\cS^{-n}\ |\  a^{(m)}\in\cF,\ 
k\in\Z\right\},
\quad
 \fF_M=\left\{\sum_{n=k}^\infty a^{(n)}\cS^{n}\ |\  a^{(m)}\in\cF,\ 
k\in\Z\right\}.
\]
\end{Def}

%\begin{Rem}
These sets equipped with addition and composition rules for difference operators have structure of skew fields \cite{CMW18a}.
The skew field  $\fF$ is a minimal subfield of the skew fields $\fF_L$ and $\fF_M$,
containing $\cFS$. 
A rational operator can be expanded in the Maclaurin or Laurent series and thus represented by an element of $\fF_M$ or $\fF_L$ respectively. 
The skew fields $\fF_L$ and $\fF_M$ are isomorphic. The 
isomorphism is given by the reflection map $\cT$.
%\end{Rem}

For $L\in \fF_L$, we denote its part with non-negative powers of $\cS$ by $L_+$, which is a difference operator.

A  rigorous  theory pseudo-difference operators with detail proofs and applications to integrable differential-difference equations can be found in  \cite{CMW18a}.

\section{Perturbative symmetry approach}\label{sec3}

With an evolutionary differential-difference equation
\begin{equation}\label{eqF}
 u_t=F(u _n,\ldots , u _m),\qquad F(u _n,\ldots , u _m)\in\cF 
\end{equation}
we associate the evolutionary derivation $X_F\in\DerF$ (the  vector field corresponding to the dynamical system (\ref{eqF})). Thus, there is a bijection between evolutionary derivations of $\cF$ and differential-difference equations.

Derivation $X_F$ enables us to differentiate any element  $a \in\cF$ in the direction of $F$. It follows from the chain rule that 
(\ref{eqF})
\begin{equation}
 \label{at}
 a_t=a_*[F]=X_F(a).
\end{equation}

In what follow we will study {\em generators of infinitesimal symmetries} of evolutionary equations and for brevity will call them symmetries. 
\begin{Def}\label{defsym}   We say that $G(u_p,\ldots,u_q)\in\cF$ is a {\em symmetry} of  
(\ref{eqF}) if $[G,F]=0.$
\end{Def}
If $G$ is a symmetry of  (\ref{eqF}), equation (\ref{eqF}) is invariant ${\rm mod}\,\epsilon^2$ under the near identity transformation $\hat{ u}=u+\epsilon G$. The evolution equation associated to this symmetry is
\begin{equation}\label{symG}
 u_\tau=G(u_p,\ldots,u_q),
\end{equation}
which is compatible with  (\ref{eqF}). These can be used as equivalent definitions for symmetry.

%Obviously $G=F$ is a symmetry. 
Let $G_1$ and $G_2$ be any two symmetries of equation (\ref{eqF}).  It follows immediately from the Definition \ref{defsym} and the Jacobi identity that the Lie bracket   $G_3 =[G_1,G_2]$ is also a symmetry of (\ref{eqF}). Thus, symmetries of equation (\ref{eqF}) form a  subalgebra of the Lie algebra $\DerF$ which will be denoted $\fA_F$
\[
 \fA_F\defeq\{G\in\cF\;|\; [F,G]=0\}.
\] 
The existence of an infinite dimensional commutative Lie algebra of symmetries is a characteristic property of integrable systems.
\begin{Def}\label{defint}
 A differential-difference equation $u_t=F$ (\ref{eqF}) is called integrable if its Lie  algebra of symmetries $\fA_F$ is infinite dimensional and contains symmetries of arbitrary high total order. 
\end{Def}

An infinite hierarchy of commuting symmetries of an integrable differential-difference  equation (\ref{eqF})  can be constructed using a recursion operator $\Lambda$ (if $\Lambda$ is known), namely, $ G_k=\Lambda( G_{k-1}  )$, where  $\Lambda\in\fF$ is a rational pseudo-difference operator with coefficients in $\cF$ satisfying the equation
\begin{equation}\label{lamt}
 X_F(\Lambda)-[F_*,\Lambda]=0.
\end{equation}

\begin{Ex}\label{ex1}
The Volterra  equation
\[
 u_t=F,\ \ F=u_1 u-u u_{-1},\ \ F_*=u\cS+u_1-u_{-1}-u\cS^{-1}
\]
has order $\ord(F)=(-1,1)$ and total order $\Ord (F)=2$. Its recursion operator 
\[
 \Lambda=AB^{-1},\qquad A=u(\cS +1)(u \cS-\cS^{-1} u),\  B=u 
(\cS-1)
\]
generates the infinite hierarchy of symmetries $G_k=\Lambda^k(F)$
\begin{eqnarray}\nonumber
 G_1&=&\Lambda(u_1 u-u u_{-1})=A(1-\cS)^{-1}(u_1-u_{-1})=
 A(u+u_{-1}+\gamma_1)\\ \label{G1} 
 &=&u  (u_1 u_2+  u_1^2+u  u_1-u  u_{-1}-  u_{-1}^2-u_{-2}   u_{-1})+\gamma_1 F,%\\   &&\nonumber
\\ \label{G2} 
 G_2&=&\Lambda^2(u_1 u-u u_{-1})=u  (u_1 u_2 u_3+ u_1 u_2^2+2  u_1^2 u_2+u  u_1 u_2+  u_1^3+2 u  u_1^2+u^2  u_1\\
 &&-u ^2 u_{-1}-2 u u_{-1}^2- u_{-1}^3-u u_{-1} u_{-2}-2 u_{-1}^2 u_{-2}-u_{-1}u_{-2}^2   -u_{-1} u_{-2}  u_{-3})+\gamma_1 G_1+\gamma_2 F,\ldots\nonumber
\end{eqnarray}
where $\gamma_k\in\k$ are arbitrary constants ($\C={\rm Ker} (B)$).
\end{Ex}
To find a recursion operator for a given equation is a difficult problem. There is a regular way to solve it if the equation has a Lax representation. A discussion of this problem and many explicit examples of pseudo-difference recursion operators can be found in \cite{CMW18a, kmw13}. 

\subsection{Testing for integrability: symmetry approach}\label{sec31}
The goal of the symmetry approach is to find necessary conditions for the existence of an infinite dimensional Lie algebra of symmetries for a given equation. Originally, it was proposed and developed by A.B. Shabat and his team, and applied to study of partial differential equations \cite{mr86i:58070}-\cite{mr93b:58070}. The formalism enables them to develop an explicit test for integrability and solve a number of classification problems by producing complete lists of integrable partial differential equations.  Later on the symmetry approach has been extended to some non-evolutionary PDEs \cite{mnw07}, integro-differential equations \cite{mn1, mn2} and recently to partial-difference and differential-difference equations \cite{MWX, adler14}. The method is inspired by the observation that the existence of an infinite algebra of symmetries implies   existence of  a formal   solution $\Lambda$ of equation (\ref{lamt}) in terms of a formal series. The conditions of solvability of the equation can be explicitly formulated and provide us with necessary conditions of integrability. 

In the case of differential difference equations, the only result in this aspect so far is obtained by Adler \cite{adler14}. %Here we quote this important statement. 

\begin{Thm}[Adler \cite{adler14}]\label{Thad}  
\noindent
If an evolutionary differential-difference equation
\begin{equation}\label{eqFF}
 u_t=F(u _n,\ldots , u _m),\qquad F(u _n,\ldots , u _m)\in\cF ,
\end{equation}
admits  symmetries
\begin{equation}\label{symGG}
 u_{\tau }=G(u_{p},\ldots,u_{q})
\end{equation}
with  $q$ arbitrarily large then equation (\ref{lamt}) admits a solution $\Lambda_L\in\fF_L$ of the
form
\[
 \Lambda_L=F_*+\sum_{k=0}^\infty a^{(-k)}\cS^{-k},\qquad\  a^{(-k)}\in\cF.
\]
If equation (\ref{eqFF})   admits  symmetries (\ref{symGG}) with 
$\ -p$  
arbitrarily large then equation (\ref{lamt}) admits a solution $\Lambda_M\in\fF_M$  of the
form
\[
 \Lambda_M=F_*+\sum_{k=0}^\infty b^{(k)}\cS^{ k},\qquad\ b^{(k)}\in\cF.
\]
\end{Thm}
We call the solutions $\Lambda_L$ and $\Lambda_M$ of (\ref{lamt}) formal recursion operators (formal symmetries in \cite{adler14}) for the equation. 
Let us consider the implications from the first part of the above statement. We try to find a formal series $\Lambda_L$ satisfying equation (\ref{lamt}) for a function $F= F(u _n,\ldots , u _m)$. The substitution of $\Lambda_L, F,F_*$ in  (\ref{lamt}) results in a formal Laurent series, each coefficient of which should vanish. Collecting the coefficients at the powers  $\cS^{m-k},\ k=0,1,\ldots$ we obtain a triangular system of difference  equations to determine the coefficients $a^{(-k)}$ of $\Lambda_L$. This system is of the form:
\begin{equation}
 \label{eqak}
 \cS^{m-k}:\ \ \cM_{m,k} (a^{(-k)})=C_{m,k},\qquad 
  \qquad k=0,1,\ldots
\end{equation}
where 
\[
  \cM_{m,k}  \defeq\frac{\partial F}{\partial u_m}\cS^m - \cS^{-k}\left( \frac{\partial F}{\partial u_m}\right) 
 \]
and the terms $C_{m,k}$ depend on the function $F(u _n,\ldots , u _m) $,  its partial derivatives and   the coefficients $a^{(0)}, a^{(-1)},\ldots, a^{(1-k)}$ only. For example
\begin{eqnarray}
&&C_{m,0}=X_F\left(\frac{\partial F}{\partial u_m}\right)=\sum\limits_{i=n}^m\cS^i(F) \frac{\partial^2 F}{\partial u_{i}\partial u_{m}},\\
&&C_{m,1}=
 X_F\left(\frac{\partial F}{\partial u_{m-1}}+a^{(0)}\delta_{m,1}\right)+\frac{\partial F}{\partial u_{m-1}}\left(a^{(0)} -\cS^{m-1}(a^{(0)})\right) .
\end{eqnarray}
 According to  Theorem \ref{Thad}, for integrable equation (\ref{eqF}) a solution of the triangular system (\ref{eqak})  exists and $a^{(-k)}\in\cF$. Thus the necessary integrability conditions for equation (\ref{eqF}) are conditions of solvability in $\cF$  of the system   (\ref{eqak}) with respect to the coefficients $a^{(-k)}$.  The latter can be formulated as the conditions that $C_{m,k}$ belong to the image spaces of the linear difference operators $\cM_{m,k}$, i.e.,
 $ C_{m,k}\in \mbox{Im} \cM_{m,k},\ k=0,1,\ldots$.
 In the case $k=0$ it reduces to the problem of membership in the space $ \mbox{Im}\,  (\cS-1)$, which has a well known solution. Namely, if  $a\in\k\bigoplus\mbox{Im}\,  (\cS-1)$ then $\delta_u (a)=0$. Here $\delta _u$ is the variational derivative
 \[
  \delta_u(a)\defeq\sum_{n\in\Z}\cS^{-n}\left(\frac{\partial a}{\partial u_n}\right)=\frac{\partial}{\partial u_0}\sum_{n\in\Z}\cS^{n}(a).
 \]
There is   an algorithmic way to solve the membership problem for the spaces $\mbox{Im} \cM_{m,k}$ and if $C_{m,k}\in \mbox{Im} \cM_{m,k},\ k=0,1,\ldots$ to find the coefficients $a^{(-k)}\in\cF$ recursively.
\begin{Ex}
 Let us consider equation (\ref{eqF}) with $F=f(u) (u_1-u_{-1})$, where   $f(u)\in\cF$. In this case $m=1,\ \frac{\partial F}{\partial u_1}=f(u)$ and  
\[
 C_{1,0}=X_F(f(u) )=f'(u) f(u) (u_1-u_{-1}),\quad \cM_{1,0}=f(u) (\cS-1).
\]
The first necessary integrability condition $C_{1,0}\in \mbox{Im}\, \cM_{1,0}$ leads to
\[
 f'(u)  (u_1-u_{-1})\in\mbox{Im}\,  (\cS-1).
\]
Thus
\[
 \delta_u\left(f'(u)  (u_1-u_{-1})\right)=f''(u) u_1+\cS^{-1}(f'(u) )-f''(u) u_{-1}-\cS(f'(u))=0
\]
Taking the partial  derivative with respect to $u_1$ we get $f''(u)=\cS(f''(u))$, therefore $f''(u)$ is a constant (an element of $\k$) and thus
$f(u)=\alpha u^2+\beta u+\gamma$, where $\alpha,\beta,\gamma\in\k$ are arbitrary constants. The resulting equation
\[
 u_t=(\alpha u^2+\beta u+\gamma)(u_1-u_{-1})
\]
is known to be integrable. 
\end{Ex} 
In this example the first integrability condition enables us to give a complete description of all integrable differential-difference equations of the form $u_t=f(u)(u_1-u_{-1})$.

In Theorem \ref{Thad} the order of local formal recursion operators depends on the given equation and there is no simple way to find  expressions for the coefficients $a^{(-k)},b^{(k)}$ explicitly. In Section \ref{sec5} using symbolic representation we will proof the existence of the universal (canonical) quasi-local formal recursion operators  and present explicit recursive formulae  for its coefficients.

\subsection{Approximate symmetries and integrability}\label{secAprox}

Approximate symmetries proved to be successful in the case of partial differential and integro-differential equations \cite{mn1, mn2}. Our definition of approximate symmetries is based on the degree grading of the polynomial ring $\cR$ defined by (\ref{gru}). 

In order to define approximate symmetries let us consider an evolutionary differential-difference equations 
\begin{equation}
\label{eqR} 
u_t=F,\qquad F\in\cR. 
\end{equation}
 Then $F$ can be uniquely represented as a sum of its homogeneous components
\[
 F=F^{(0)}+\ldots +F^{(N)},
\]
where $F^{(k)}\in\cR^k$ or zero,  $N$ is the degree of the polynomial $F$. It follows from Definition \ref{defsym} that a polynomial $G=G^{(0)}+\ldots +G^{(M)}$ is a symmetry if the Lie bracket $[F,G]=0$ vanishes. Due to the degree grading structure of the Lie algebra (\ref{bigrLie}), the latter is equivalent to a sequence of homogeneous equations
\begin{equation}
 \label{FGp}
 \sum\limits_{k=0}^p [F^{(k)},G^{(p-k)}]=0,\qquad p=0,1,2,\ldots ,N+M\, .
\end{equation}
Here we assume that $F^{(k)}=0$ if $k\not\in\{0,\ldots N\}$ and $G^{(k)}=0$ if $k\not\in\{0,\ldots M\}$. If all $N+M+1$ equations are satisfied, then $G$ is a symmetry. If in the (\ref{FGp}) the first $s+2$ equations with $0\leqslant p\leqslant s+1$ are satisfied, then     
$G$ is called  {\it $s$--approximate symmetry} (or {\it approximate symmetry of degree $s$})  of the equation (we ignore the rest equations). It follows from the Jacobi identity that the Lie bracket of two $s$--approximate symmetries is also a $s$--approximate symmetry and therefore approximate symmetries of degree $s$ form a closed Lie algebra $\fA_F^s$
\[
 \fA_F^s\defeq\{ G\in\cR\;|\; \pi_k([F,G])=0\ \mbox{for}\ k=0,1,\ldots,s\}.
\]
In other words, the move from symmetries to approximate symmetries of degree $s$ is the transition from definitions and computations made in terms of the ring $\cR$ to the quotient ring $\cR\diagup \cI^{s+1}$, where 
\begin{equation}
\label{cI}
 \cI=\langle\bu\rangle\subset\cR
\end{equation}
 is the maximal ideal generated by the variables $\bu$. Namely, we say that $G\in\cR$ is 
$s$--approximate symmetry of equation (\ref{eqR}) if $[F,G]\in\cI^{s+1}$. We will also need difference operators and formal series  with coefficients in the quotient rings $\cR\diagup \cI^{s+1}$, but their definitions will be more natural in symbolic representation which will be introduced in the next section.

\begin{Rem}
There are no obstructions to replace the polynomial ring $\cR=\k [\bu]$ by the ring $ \bar{\cR}$ of formal series  
 \begin{equation}\label{rbar}                                                                                            
        \bar{\cR}=\k [[\bu]]=\{\sum_{k=0}^\infty f^{(k)}\,|\,    f^{(k)}\in\cR^k\}                                                                                          
 \end{equation}
 in variables $\bu$ and consider equations and their symmetries given by the Maclaurin expansion or just formal series. The above definition of $s$--approximate symmetries remains correct since only a finite number of equations (\ref{FGp}) 
 \[
  \pi_k ([F,G])=0, \qquad k=0,1,\ldots, s,\quad F,G\in\bar{\cR}
 \]
 has to be verified. 
\end{Rem}

\begin{Def}\label{defints}
 A differential-difference equation $u_t=F,\ F\in\bar{\cR}$ is called $s$--approximate integrable if its Lie  algebra $\fA_F^s$  is infinite dimensional and contains $s$--approximate symmetries of arbitrary high total order. The equation is called formally integrable, if it is $s$--approximate integrable for arbitrary large $s$.
\end{Def}
Compare Definitions  \ref{defint} and \ref{defints} we conclude that integrabe equations are formally integrable. The converse may not be true since we do not assume any convergence of formal series for the  equation and its formal symmetries. Conditions of $s$--approximate integrability can be explicitly written and verified. They are strong necessary conditions for integrability and proved to be suitable for classification of integrable equation.  
In next section we will formulate criteria of $s$--approximate integrability using symbolic representation, which enable us to test for integrability of a given equation as well as to progress in solution of classification problem for integrable differential-difference equations.

In what follows, we will restrict ourselves by equations $u_t=F$ and symmetries $G$ without a constant term ($\pi_0(F)=0,\ \pi_0(G)=0$). Let $\cR'$ denote a subspace of formal series without a constant term
\begin{equation}
 \label{cR_1}
\cR'\defeq (1-\pi_0)\bar{\cR}=\bigoplus\limits_{n=1}^\infty \cR^n\, .
\end{equation}
Having $N$ linearly independent symmetries containing a constant term we can always construct $N-1$ linear combinations without constant term. A constant term in $F$ can often be   removed by a simple invertible transformation of variables. Thus the restriction  to the linear space $\cR'$ (\ref{cR_1}) of formal series with no constant terms  does not really affect the generality, but considerably reduces unnecessary stipulations in each case.

\section{Symbolic representation}\label{sec4}
Symbolic representation is widely used in theory of pseudo-differential operators. It has been
first applied  in the area of integrable systems  by Gel'fand and Dickey \cite{mr58:22746} and
further developed in works of Beukers, Sanders \& Wang \cite{mr99i:35005,
mr99g:35058} and Mikhailov \& Novikov \cite{mn1, mn2}. In this section, we will extend the methods of symbolic representation 
to the ring of difference polynomials, difference operators and formal series. 
Similar to the differential case, the symbolic representation can be viewed as a simplified notation for a Fourier transform ($Z$-transform).

\subsection{Symbolic representation of difference polynomials}
To define the symbolic representation $\hat{\cR}=\oplus\hat{\cR}^n$ of the degree graded difference
ring (and Lie algebra)  $\bar{\cR}=\oplus\cR^n$ of formal series, we first define
an isomorphism of the $\k$--linear spaces $\varphi:\cR^n\mapsto \hat{\cR}^n$ and then extend it to
the difference   ring  and Lie algebra isomorphism equipping $\hat{\cR}$ with the
multiplication, derivation and Lie bracket. The
isomorphism of the $\k$--linear spaces $\varphi:\cR^n\mapsto \hat{\cR}^n$ is uniquely defined by its action
on monomials.

\begin{Def}\label{repmon}
The {\em symbolic form} of a difference monomial terms is defined as
\begin{equation*}
\varphi\,:\, \alpha \mapsto\alpha,\quad \varphi\,:\,\alpha u_{i_1}u_{i_2}\cdots u_{i_n}\in\cR^n \   \mapsto \  \hu^{n}\alpha
\langle\xi_1^{i_1} \xi_{2}^{i_2} \cdots \xi_{n}^{i_n}\rangle_{\Sg_n}\in\hring^n,\quad \alpha\in\k,
\end{equation*}
where $\langle\cdot \rangle_{\Sg_n}$ denotes the average over the permutation group
${\Sg_n}$    of $n$ elements $\xi_1,\ldots ,\xi_n$:
\[  \langle a(\xi_1, \cdots ,\xi_{n})\rangle_{\Sg_n}= \frac{1}{n!}
\sum_{\sigma\in \Sg_n}
a(\xi_{\sigma (1)}, \cdots , \xi_{\sigma(n)}).
\]
\end{Def}
Notations $\xi_1,\xi_2,\ldots$ are reserved for the variables in symbolic representation. We will omit ${\Sg_n}$ in the group average $\langle\cdot \rangle_{\Sg_n}$, assuming that the average is taken over the permutation group of $n$ variables $\xi_k$ where $n$ is shown as the degree of $\hat{u}^n$.

A few examples:
\begin{eqnarray*} &&u_k\raph \hu\xi_1^k,\ \ \ u^n\raph \hu^n,\ \ \ 
u_1u_2\raph\hu^2\langle\xi_1 \xi_2^2\rangle_{\Sg_2}=\frac{\hu^2}{2}(\xi_1 \xi_2^2+\xi_1^2 \xi_2),\\&&
\alpha uu_p^2+\beta u_q^3 \raph\hu^3 \left(\frac{\alpha}{3}(\xi_2^p \xi_3^p+\xi_1^p \xi_2^p+\xi_1^p \xi_3^p)+
\beta \xi_1^q \xi_2^q \xi_3^q\right),\ \ \alpha,\beta\in\k.\end{eqnarray*}

With this isomorphism a homogeneous polynomial $f\in\cR^n$ is in one-to-one correspondence with a term  $\hat{f}=\hu^n
a(\xi_1,\ldots,\xi_n)\in\hring^n $ in which the coefficient function $a(\xi_1,\ldots,\xi_n)$ is \(n\)-variable symmetric Laurent polynomials, i.e.,
$a(\xi_1,\ldots,\xi_n)\in\Xi_n,$ where $\Xi_n\defeq\k[\xi_1 ,\xi_1^{- 1},\ldots,\xi_n,\xi_n^{- 1}]^{\Sigma_n}$. 

The projector $\pi_k$ (\ref{pi}) selects the $k$-th homogeneous component of an element $f\in\bar{\cR},\ \pi_k(f)\in\cR^k$. Its symbolic representation $\hat{\pi}_k$ is induced by the condition $\hat{\pi}_k\varphi=\varphi\pi_k$. Let 
\[f=\sum_{n\geqslant0}f^{(n)}\in\bar{\cR},\quad  \pi_k(f)=f^{(k)}\in\cR^k,\quad \varphi(f^{(n)})=\hat{u}^n a_n(\xi_1,\ldots,\xi_n),\] 
then 
\[ \hat{\pi}_k\left(\sum_{n\geqslant0} \hat{u}^n a_n(\xi_1,\ldots,\xi_n)\right)=\hat{u}^k a_k(\xi_1,\ldots,\xi_n).\]

The action of automorphisms $\cS,\cT$ in symbolic representation is given by
\iffalse

the linear spaces $\hring^n$ corresponding to $\cR^n$ have the
property that the coefficient functions $a(\xi_1,\ldots,\xi_n)$ of symbols $$\hu^n
a(\xi_1,\ldots,\xi_n)\in\hring^n $$ are \(n\)-variable symmetric Laurent polynomials, i.e.,
$a(\xi_1,\ldots,\xi_n)\in\k[\xi_1^{\pm 1},\ldots,\xi_n^{\pm 1}]$.

Let $f\in\cR^n$ and
$\varphi\,:\,f\mapsto \hu^n a(\xi_1,\ldots,\xi_n)$ then
\fi
\[ \cS (f)\quad \raph\quad \hu^n a(\xi_1,\ldots,\xi_n)(\xi_1\cdots \xi_n)\, ,\quad
\cT (f)\quad\raph \quad \hu^n a(1/\xi_1,\ldots,1/\xi_n)\, ,
\]
and thus $\varphi\,:\,\cS^k (f)\ \mapsto \ \hu^n a(\xi_1,\ldots,\xi_n)(\xi_1\cdots
\xi_n)^k$.

\begin{Def}
 Let $f\in\cR^n,\ \varphi(f)= \hu^n a(\xi_1,\ldots,\xi_n)$ and $g\in\cR^m,\ \varphi(g=
\hu^m b(\xi_1,\ldots,\xi_m)$, then
 $\varphi(f g)=\varphi(f)\star\varphi(g)$ where
\begin{equation}\label{starprod}
\varphi(f)\star\varphi(g)=\hu^{n+m}\langle a(\xi_1,\ldots,\xi_n)
b(\xi_{n+1},\ldots,\xi_{n+m})\rangle_{\Sg_{n+m}}.
\end{equation}
\end{Def}

If $f\in\cR^0=\k$, then   $ \varphi(fg)=f\varphi(g)$.
The representation of
difference monomials (Definition \ref{repmon}) can be deduced from $\varphi: u_k\mapsto
\hu \xi_1^k$ and this multiplication rule (\ref{starprod}).

The linear space $\hat{\cR}=\oplus\hat{\cR}^n$  equipped with the $\star$ multiplication is isomorphic to the graded ring $\bar{\cR}$.

\subsection{Difference operators in the symbolic representation}\label{doisr}
We assign the symbol $\eta$ corresponding to the shift operator $\T$ in the symbolic representation with the action
$$
\eta(\hu^pa(\xi_1,\ldots,\xi_p))=\hu^p \xi_1\xi_2\cdots\xi_p a(\xi_1,\ldots,\xi_p).
$$
and the composition rule (corresponding to $\cS\circ f=\cS(f) \cS$):
\[
\eta\circ \hu^n a(\xi_1,\ldots , \xi_n)=\hu^n(\prod_{j=1}^{n}\xi_j) a(\xi_1,\ldots ,
\xi_n)\eta \, ,
\]
and thus $\eta^k\circ \hu^n a(\xi_1,\ldots , \xi_n)=\hu^n(\prod_{j=1}^{n}\xi_j)^k a(\xi_1,\ldots ,
\xi_n)\eta^k$.

Now we can extend the symbolic representation to difference operators with coefficients in $\bar{\cR}$
\begin{equation}\label{phiA}
 \varphi\,:\, A=\sum\limits_{k=p}^q f_{(k)}\cS^k\ \mapsto  \hat{A}=\sum\limits_{k=p}^q \varphi (f_{(k)})\eta^k,\qquad f_{(k)}\in\bar{\cR}.
\end{equation}
Here the coefficients $f_{(k)}$ are  formal series 
$$f_{(k)}=\alpha_k+\sum_{n=1}^\infty f_{(k)}^{(n)},\qquad \alpha_k\in\k,\quad f_{(k)}^{(n)}\in\cR^n,$$
and thus
\[
 \varphi (f_{(k)})=\alpha_k+\sum_{n=1}^{\infty}\hat{u}^n A_{k,n}(\xi_1,\ldots,\xi_n),\quad \ \ A_{k,n}(\xi_1,\ldots,\xi_n)\in\Xi_n.
\]
After substitution of $ \varphi (f_{(k)})$ in (\ref{phiA}) and the change of  the order of summation, the 
 symbolic representation operator $\hat{A}=\varphi(A)$ can also be written in the form
\begin{equation}\label{Aa}
 \hat{A}=\sum\limits_{n=0}^{\infty} \hat{u}^n A_n(\xi_1,\ldots,\xi_n,\eta)
,\qquad 
 A_n(\xi_1,\ldots,\xi_n,\eta)=\sum\limits_{k=p}^q A_{k,n}(\xi_1,\ldots,\xi_n)\eta^k\, .
\end{equation}
%where $N=\max (N_p,N_{p+1},\ldots,N_q)$. 
In the non-commutative ring of difference operators in the symbolic representation there is a natural degree grading (in powers $\hu^n$). The projector $\hat{\pi}_s$ on the homogeneous component can be extended to operators (\ref{Aa})
\[
 \hat{\pi}_s (\hat{A})=\hat{u}^s A_s(\xi_1,\ldots,\xi_s,\eta).
\]

Symbolic representation enables us to solve immediately the problem to find the pre-image of the Fr\'echet derivative. For example, let $F=u(u_1-u_{-1})$, then 
$F_*=u\cS+(u_1-u_{-1})-u\cS^{-1}$ and
\[
 \varphi(F)=\hu^2 \frac{1}{2}(\xi_1+\xi_2-\xi_1^{-1}-\xi_2^{-1}),\quad \varphi(F_*)=\hu (\xi_1+\eta-\xi_1^{-1}-\eta^{-1})
\]
In general for $f\in\cR^n$ we have a symbol $\varphi(f)=\hu^n a (\xi_1,\ldots,\xi_n)$, where $a$ is a symmetric Laurent polynomial, and
\[  f_*\raph  \left(\hu^n a (\xi_1,\ldots,\xi_n)\right)_*=    n \hu^{n-1}a(\xi_1,\ldots ,\xi_{n-1},\eta )\, .\]
Similar to the differential polynomial case \cite{mn1},  the symbol of the Fr\'echet derivative is always
symmetric with respect to all permutations of its variables, including the variable
$\eta$. 
\begin{Pro}\label{profre} An operator (\ref{Aa}) is a symbolic representation of the Fr\'echet derivative of an element
of $\bar{\cR}$ if and only if each term $\hat{u}^n A_n(\xi_1,\ldots,\xi_n,\eta)$ is a Laurent polynomial in its variables and 
is invariant with respect to
all permutations of its variables, including the variable $\eta$.
\end{Pro}
Given the Fr\'echet  derivative $\hat{f_*}$ in the symbolic representation, we can immediately reconstruct $\varphi(f)+\alpha$ by replacing $\hu^{n-1}\rightarrow \hu^{n}/n$ and $\eta\rightarrow\xi_{n}$ in each term, where $\alpha\in\k$ is an arbitrary constant.

The composition rule for difference operators in symbolic representation (\ref{Aa}) follows from 
\begin{equation}
 \label{comprule}
 \hat{u}^n A_n(\xi_1,\ldots,\xi_n,\eta)\circ \hat{u}^m B_m(\xi_1,\ldots,\xi_m,\eta)=
 \hat{u}^{n+m}\langle A_n(\xi_1,\ldots,\xi_n,\eta\prod\limits_{i= 1}^{ m}\xi_{i+n})B_m(\xi_{n+1},\ldots,\xi_{n+m},\eta)\rangle
\end{equation}
and the linearity. In particular, the action of difference operators on elements of $\hat{\cR}$ follows from
\begin{equation}
 \label{actrule}
 \hat{u}^n A_n(\xi_1,\ldots,\xi_n,\eta)( \hat{u}^m b_m(\xi_1,\ldots,\xi_m))=
 \hat{u}^{n+m}\langle A_n(\xi_1,\ldots,\xi_n, \prod\limits_{i= 1}^{ m}\xi_{i+n})b_m(\xi_{n+1},\ldots,\xi_{n+m})\rangle.
\end{equation}

%The symbolic representation (\ref{Aa}) respects the degree grading structure of the ring $\bar{\cR}$ and $\hat{\pi}_n(\hat{A})=\hat{u}^n A_n(\xi_1,\ldots,\xi_n,\eta)$. 
%In the case of the ring $\bar{\cR}$ of formal series in $\bu$ (the $\cI$--adic completion  of $\cR$) we just replace  $N$ by $\infty$ in (\ref{Aa}). 

Now we are going to define the action of evolutionary derivations on elements of the ring and difference operators in symbolic representation. 
A differential-difference equation 
\[
u_t=f,\qquad f=\sum_{m\geqslant0} f^{(m)},\qquad f^{(m)}\in\cR^m
\]                        
defines the evolutionary derivation $\partial_t=X_f=\sum_{m\geqslant0} X_{f^{(m)}}$. Let $\varphi(f^{(m)})=\hu^m b_m(\xi_1,\ldots,\xi_m)$ and for an element $g \in\cR^n,\ 
\hat{g}=\varphi(g)=\hu^n a(\xi_1,\ldots,\xi_n)$.  Then
\[
 \hat{g}_t\defeq\varphi(X_f(g))=\sum\limits_{m=0}^\infty 
 \hu^{n+m-1} n \langle a(\xi_1,\ldots,\xi_{n-1},\prod_{i=0}^{m-1}\xi_{n+i})b_m(\xi_n,\ldots, \xi_{n+m-1})\rangle_{\Sigma_{n+m-1}}\, .
\]
Similarly for a difference operator $\hat{A}=\varphi(A)$ (\ref{Aa}) we obtain
\begin{equation}\label{hatAt}
\hat{A}_t\defeq \varphi(X_f(A))=
 \sum\limits_{m=0}^\infty \left(\sum\limits_{n=1}^\infty 
 \hu^{n+m-1} n \langle A_n(\xi_1,\ldots,\xi_{n-1},\prod_{i=0}^{m-1}\xi_{n+i},\eta)b_m(\xi_n,\ldots, \xi_{n+m-1})\rangle_{\Sigma_{n+m-1}}\right)\, .
\end{equation}

Finally we consider the symbolic representation of the  Lie bracket. 
  Let $f\in\cR^n,\ f\raph \hu^n a(\xi_1,\ldots,\xi_n)$ and $g\in\cR^m,\ g\raph
\hu^m b(\xi_1,\ldots,\xi_m)$, then in symbolic representation the Lie bracket $[f,g]$  (\ref{LieF}) takes form
\begin{eqnarray}
[f,\ g] & \raph  &\hu^{n+m-1}
\left\langle  m b(\xi_1,\ldots,\xi_{m-1},\xi_m\cdots \xi_{n+m-1})
a(\xi_{m},\ldots,\xi_{n+m-1})
\right.\nonumber \\
&&\left.  -n a(\xi_1,\ldots,\xi_{n-1},\xi_n\cdots
\xi_{n+m-1})b(\xi_{n},\ldots,\xi_{n+m-1}) \right\rangle _{\Sg_{n+m-1}}\label{slie}
\end{eqnarray}
In particular, if $f\in\cR^1,\ f\raph \hu\omega(\xi_1)$ and $g\in\cR^m,\ g\raph
\hu^m b(\xi_1,\ldots,\xi_m)$, then
\begin{equation}\label{lincomm}
[f,g] \raph -\hu^m\left(\omega(\xi_1\cdots \xi_m)-\omega(\xi_1)-\cdots -\omega(\xi_m)\right)\
 b(\xi_1,\ldots,\xi_m).
\end{equation}
%which equals to zero only if $m=1$. Thus we have the following result:
\begin{Pro}\label{pro1}
 Let $u_t=F$ be a linear equation ($F\in\cR^1$), then its linear space of symmetries coincides with $\cR^1$. 
\end{Pro}
\begin{proof}
 Let us assume that $G=G^{(1)}+G^{(m)}+G^{(m+1)}+\cdots$  for some $m>1$, where $ G^{(k)}=\pi_k(G)$ is a symmetry, i.e. $[F,G]=0$ and thus $[F,G^{(m)}]=0$. It follows from (\ref{lincomm})  with
$f=F,\ g=G^{(m)}$ that $b(\xi_1,\ldots,\xi_m)=0$, and therefore $G^{(m)}=0$.
\end{proof}

\begin{Pro}\label{prolin}
 Let $G\in\cR'$ be a symmetry of equation $u_t=F\in\cR'$  with a non-zero linear term $F^{(1)}=\pi_1(F)\ne 0$. Then $G$ also has a non-zero linear term $\pi_1(G)\ne 0$.
\end{Pro}
\begin{proof}
 Let us assume that $G=G^{(m)}+G^{(m+1)}+\cdots$ and $G^{(m)}\ne0$ for some $m>1$. Thus $\pi_m([F,G])=[F^{(1)},G^{(m)}]=0$. A contradiction follows from (\ref{lincomm}) with
$f=F^{(1)},\ g=G^{(m)}$, so  that $b(\xi_1,\ldots,\xi_m)=0$, and therefore $G^{(m)}=0$.
\end{proof}

\begin{Pro}
Any two symmetries $G_1,G_2\in\cR'$    of equation $u_t=F\in\cR'$  with a non-zero linear term $F^{(1)}=\pi_1(F)\ne 0$ commute $[G_1,G_2]=0$.
\end{Pro}
\begin{proof}
 The commutator of two symmetries $H=[G_1,G_2]$ is a symmetry and it does not have a linear part $\pi_1(H)=0$. It follows from Proposition \ref{prolin} that $H=0$.  
\end{proof}

\subsection{Symmetries in the symbolic representation.}
Symbolic representation enables us to reduce the problem of description of symmetries for a given differential-difference equation
\begin{equation}\label{Fsymb}
u_t=F,\qquad F \raph \hu\omega(\xi_1)+\hu^2 a_2(\xi_1,\xi_2)
+\hu^3a_3 (\xi_1,\xi_2,\xi_3)+\cdots, \quad \omega(\xi_1)\neq 0%{\rm const} ,
\end{equation}
to the multivariate polynomial factorisation  problem.  
\begin{Thm}\label{theorsym}
 Let $G\in\cR'$ be a symmetry of equation (\ref{Fsymb}). Then the coefficients $A_k$ of its symbolic representation
\begin{equation}\label{Gsymb}
 G\raph \hu\Omega(\xi_1)+\hu^2 A_2(\xi_1,\xi_2)
 +\hu^3 A_3(\xi_1,\xi_2,\xi_3)+\cdots
\end{equation}
can be determined recursively:
\begin{eqnarray}
&& A_2(\xi_1,\xi_2)=\frac{ G^\Omega (\xi_1,\xi_2)} {G^\omega
(\xi_1,\xi_2)}a_2(\xi_1,\xi_2);\label{A2}\\
&&A_{m}(\xi_1,...,\xi_{m})=\frac{1}{G^\omega (\xi_1,...,\xi_{m})}\bigg(G^\Omega
(\xi_1,...,\xi_{m})a_{m}(\xi_1,...,\xi_{m})\nonumber\\
&&\quad\qquad + \sum_{j=2}^{m-1} j\left\langle  A_j(\xi_1,...,\xi_{j-1}, \xi_j \ldots \xi_m)
a_{m+1-j}(\xi_{j},...,\xi_{m})\right.\nonumber\\
&&\quad \qquad 
\left. - a_{j} (\xi_1,...,\xi_{j-1},\xi_j \ldots \xi_m)
 A_{m+1-j}(\xi_{j},...,\xi_{m}) \right\rangle
_{\varSigma_{m}}\bigg),\quad m=3, 4, \cdots \label{As}
\end{eqnarray}
where
\begin{equation}\label{Gw}
G^\varkappa (\xi_1,...,\xi_m)=\varkappa(\prod_{i=1}^{m}\xi_i)-\sum_{i=1}^{m}\varkappa(\xi_i), \quad \varkappa=\omega, \Omega.
\end{equation}
\end{Thm}

\begin{proof}
 The proof of the Theorem is straightforward. Using (\ref{slie}) we can compute the Lie bracket between $F$ and
$G$. It is obvious that its linear part vanishes. When the Lie bracket vanishes up to $\hring^2$, we have
$$
\hu^2\bigg(\omega(\xi_1\xi_2)A_2(\xi_1,\xi_2)+a_2(\xi_1,\xi_2)(\Omega(\xi_1)+\Omega(\xi_2))-\Omega(\xi_1\xi_2)a_2(\xi_1,\xi_2)-A_2(\xi_1,\xi_2)(\omega(\xi_1)+\omega(\xi_2))\bigg)=0,
$$
which leads to the expression of $A_2(\xi_1,\xi_2)$ as (\ref{A2}). The Lie bracket vanishing up to
$\hring^{m}$ is equivalent to formula (\ref{As}).
\end{proof}

Theorem \ref{theorsym} states that a symmetry $G$ of equation (\ref{Fsymb}) is uniquely determined by its linear part $\Omega(\xi_1)$.  For a given Laurent polynomial $\Omega(\xi_1)$ all coefficients $A_{n}(\xi_1,...,\xi_{n})$ in the formal
series (\ref{Gsymb}) can be found recursively. 
It does not mean that any evolutionary equation has a symmetry. The terms in (\ref{Gsymb})
must represent symbols of difference polynomials, i.e. the coefficients $A_{m}(\xi_1, \ldots , \xi_{m})$ must be symmetric Laurent polynomials. In general, the coefficients $A_k$, presented in Theorem \ref{theorsym}  are symmetric rational functions (\ref{A2}), (\ref{As})  -- they have denominators $G^\omega$ (except $\omega(\xi_1)={\rm const}\neq 0$, see remark below). 
In order to define  symbols
of  difference polynomials, these denominators must cancel with appropriate factors in the
numerators. Factorisation properties of the Laurent polynomials $G^\Omega$ (\ref{Gw}) impose constraints on possible choices of  $\Omega(\xi_1)$. We call a linear term $\hu\Omega(\xi_1),\,\,\Omega(\xi_1)\in\Xi_1$ {\it admissible} for the equation with the symbolic representation (\ref{Fsymb}) if it is the linear term of a symmetry (\ref{Gsymb}). 
A linear combination    of symmetries is again a symmetry, thus the set of all admissible linear terms forms a vector space over $\k$. We denote this space by $V_F$. For integrable equations (Definition \ref{defint}) the algebra of its symmetries  and the vector space $V_F$ are infinite dimensional.

If the coefficients $A_k,\ k=2,3,\ldots,s$ (\ref{A2}), (\ref{As}) are Laurent polynomials, then $G$ (\ref{Gsymb}) is $s$--approximate symmetry of the equation. For equations with a polynomial function $F$  the sequence of the coefficients $A_m$ usually truncates, resulting in a polynomial symmetry (Example \ref{ex4}).   
\iffalse
Thus factorisation properties of polynomials $G^\omega\) and \(G^\Omega$
are crucial for the structure of the symmetry algebra of the equation, which is determined by 
the Laurent polynomials $\Omega(\xi_1)$. 
\fi

\begin{Rem} \label{omegaconst}
If $\omega(\xi_1)$ is a nonzero constant $\omega(\xi_1)=\alpha\ne 0$, then $G^\omega(\xi_1, \cdots,\xi_m)=(1-m)\alpha$ is also a non-zero constant. In this case it follows from Theorem \ref{theorsym} that all coefficients  $A_{m}(\xi_1, \ldots , \xi_{m}),\ m=2,3,\ldots$ are symmetric Laurent polynomials and therefore equation 
\begin{equation}\label{equ}
 \hu _t=\alpha \hu+\hu^2 a_2(\xi_1,\xi_2)
+\hu^3a_3 (\xi_1,\xi_2,\xi_3)+\cdots,\qquad \alpha\in\C^*
\end{equation}
admits a formal symmetry (\ref{Gsymb}) for any choice of a Laurent polynomial $\Omega(\xi_1)$. Therefore equation (\ref{equ}) is formally integrable for any choice of the coefficients $a_k(\xi_1,\ldots,\xi_k)\in\Xi_k$. It is not surprising, since there exist formal and formally invertible change of variables
\[
 \hv= \hu+\hu^2 b_2(\xi_1,\xi_2) +\hu^3b_3 (\xi_1,\xi_2,\xi_3)+\cdots,\qquad b_k(\xi_1,\ldots,\xi_k)\in\Xi_k,
\]
such that in terms of the new variable $\hv$ equation (\ref{equ}) and its symmetry become linear 
\[
 \hv _t=\alpha \hv,\qquad \hv _\tau=\Omega(\xi_1)\hv,
\]
and thus integrable.
\end{Rem}

\begin{Ex}
Consider the Narita-Itoh-Bogoyavlensky equation \cite{bogo}
\begin{equation}\label{NIB}
u_{t}=f=u\sum_{i=1}^n(u_i-u_{-i}) 
\end{equation}
The right hand side of the equation does not contain a linear term, but it can be created by a shift $u_i\to u_{i}+1,\,\,i\in\Z$, and so we can consider
\begin{equation}\label{bogo}
u_{t}=\sum_{i=1}^n(u_i-u_{-i})+u \sum_{i=1}^n(u_i-u_{-i}). 
\end{equation}
The symbolic representation of the equation is
$$
\hu_t=\hu\omega(\xi_1)+\hu^2a_2(\xi_1,\xi_2),
$$
$$
\omega(\xi_1)=P(\xi_1)-P(\xi_1^{-1}),\quad P(\xi_1):=\xi_1^n+\xi_1^{n-1}+\cdots+1=\frac{\xi_1^{n+1}-1}{\xi_1-1},
$$
$$
a_2(\xi_1,\xi_2)=\frac{1}{2}(P(\xi_1)+P(\xi_2)-P(\xi_1^{-1})-P(\xi_2^{-1})).
$$
It is known  that equation (\ref{bogo}) possesses a symmetry with linear term of the form $\hu\Omega_k(\xi_1)$ \cite{wang12}, where
$$
\Omega_k(\xi_1)=(P(\xi_1))^k-(P(\xi_1^{-1}))^k,\quad k=2,3,\ldots
$$
that is, $\hu \Omega_k(\xi_1)\in V_f$ for $k\in \mathbb{N}$.
%We conjecture that for the N-th Bogoyavlensky equation $V_f=\mbox{Span}\{\hu\Omega_2(\xi_1),\hu\Omega_3(\xi_1),\ldots\}$.
\end{Ex}
Taking an admissible linear term for a given equation, one can use Theorem \ref{theorsym} to determine the symmetry starting with this linear term. 
\begin{Ex}\label{ex4}
  Let us consider the Volterra equation
\begin{equation}
\label{NIB1}
u_{t}=u (u_{1}-u_{-1})\in \cR^2
\end{equation}
In order to introduce a linear term to this equation we make a change of variables $u_i\to u_i+1$:
\begin{equation*}
\label{vch}
u_{t}=u_{1}-u_{-1}+u(u_{1}-u_{-1})
\end{equation*}
In the symbolic representation the equation can be written as
$
\hu_{t}=\hu\omega(\xi_1)+\hu^2 a_2(\xi_1,\xi_2),
$
where
\[
\omega(\xi_1)=\xi_1-\frac{1}{\xi_1},\quad a_2(\xi_1,\xi_2)=\frac{1}{2} \left(\xi_1+\xi_2-\frac{1}{\xi_1}-\frac{1}{\xi_2}\right)
\]
We compute a symmetry starting with $\Omega(\xi_1)=\xi_1^2-\frac{1}{\xi_1^2}\in \hat{\cR}_1$. Its quadratic terms are
\begin{eqnarray*}
 A_2(\xi_1,\xi_2)=\frac{{G^\Omega (\xi_1,\xi_2)}} {G^\omega
(\xi_1,\xi_2)}a_2(\xi_1,\xi_2)
=\frac{1}{\xi_1 \xi_2}(1+\xi_1)(1+\xi_2)(\xi_1 \xi_2+1)a_2(\xi_1,\xi_2)
\end{eqnarray*}
For cubic terms we then obtain
\[
A_3(k_1,k_2,k_3)=\frac{1}{\xi_1^2 \xi_2^2 \xi_3^2}(\xi_1 \xi_2 \xi_3-1)
 \left\langle \xi_1^2\xi_2 + \xi_1 \xi_2 \xi_3^2 + \xi_1 \xi_2^2 \xi_3^2 + \xi_1 \xi_2^2 \xi_3^3\right\rangle
_{\varSigma_{3}}
\]
and all terms of degrees higher than $3$ vanish.
In standard variables this symmetry is
\begin{eqnarray*}
\nonumber
&&u_{\tau}=u_2-u_{-2}+u u_2+u u_1+u_1^2+u_1u_2-u_{-1}u_{-2}-uu_{-2}-u_{-1}^2-u u_{-1}\\ 
&&\quad +u^2u_1+u u_1^2-u u_{-1}^2-u_{-1}u ^2+u u_1u_2-u u_{-1}u_{-2}.
\end{eqnarray*}
After changing variable $u_i\to u_i-1$, we get the symmetry for the Volterra chain (compare with $G_1$ in  Example \ref{ex1})
\begin{eqnarray*}
&&u_{\tau}=u (u_1u_2+u_1^2+u u_1-u u_{-1}-u_{-1}^2-u_{-1}u_{-2})-4u (u_1-u_{-1}).
\end{eqnarray*}
\end{Ex}

Theorem \ref{theorsym} can be used as a  test for integrability if $\Omega(\xi_1)$ is  assumed.
In this case the integrability conditions are conditions of Laurent polynomiality of the coefficient functions $A_2(\xi_1,\xi_2), A_3(\xi_1,\xi_2,\xi_3),$ etc. A negative result (i.e. the coefficient functions fail to be Laurent polynomials) might be inconclusive if the assumption about $\Omega(\xi_1)$ was wrong. 

Most interesting integrable systems possess an infinite hierarchy of local conservation laws. Let us recall that a difference polynomial (or a formal series)  $\rho$ is a density of a local conservation law for the equation $u_t=F$ (\ref{Fsymb}) if $\rho_t \in(\cS-1) \cR'$. In order to exclude trivial densities, i.e. elements of $(\cS-1) \cR'$, the densities are defined on the quotient space (the $\k$-- linear space of functionals) $\rho\in\cR'\diagup (\cS-1) \cR'$. In the symbolic representation the condition that a term $\hu^k a_k(\xi_1,\ldots,\xi_k)$ is in the image of $(\cS-1)$  means that $a_k(\xi_1,\ldots,\xi_k)$ can be presented as a product of $\xi_1\xi_2\cdots \xi_k-1$ and a Laurent polynomial.

The existence of local conservation laws imposed constraints on the linear part of equations and their symmetries.

\begin{Pro}\label{pro2}
Suppose equation $u_t=F$ (\ref{Fsymb}) possesses a conserved density $\rho$ without a linear term, then 
$\omega(\xi)+\omega( \xi^{-1})=0$. 
%If $G$ (\ref{Gsymb}) is a symmetry of the equation, then $\Omega(\xi)+\Omega( \xi^{-1})=0$.
\end{Pro}
\begin{proof}
Let the conserved density $\rho=\rho^{(k)}+\rho^{(k+1)}+\cdots$  have a nontrivial contribution  with  lowest degree    $k\geq 2$ and   $\rho^{(k)}\raph\hu^k a(\xi_1, ..., \xi_k)$, where the coefficient $a(\xi_1, ..., \xi_k)$ is not divisible by $\xi_1\xi_2\cdots \xi_k-1$ and a Laurent polynomial (otherwise $\rho^{(k)}\in \mbox{Im} (S-1)$). 
 Then    
 \[ X_{F^{(1)}}(\rho^{(k)})\in(\cS-1)\cR',\qquad X_{F^{(1)}}(\rho^{(k)})\ \raph\ \hu^k a(\xi_1, ..., \xi_k) \left(\omega(\xi_1) +\cdots \omega(\xi_k)\right),\] 
 and  should exists a Laurent polynomial $b\in\Xi_k$ such that  $\omega(\xi_1) +\cdots \omega(\xi_k)=(\xi_1\xi_2\cdots \xi_k-1)b(\xi_1, ..., \xi_k)$.  Setting $\xi_1=\ldots=\xi_k=1$ in the latter equation we find $k\omega(1)=0$, and thus $\omega(1)=0$. 
Choosing now $\xi_1=\xi,\xi_2=\xi^{-1}$ and $\xi_\ell=1$ for $\ell>2$ we have $\omega(\xi)+\omega(\xi^{-1})+(k-2)\omega(1)=0$, and therefore $\omega(\xi)+\omega(\xi^{-1})=0$.
\end{proof}
\begin{Rem}
 If an evolutionary equation (\ref{Fsymb}) possesses at least two conserved
densities in $\cR'$, then we can omit the condition ``{\em possesses a conserve density $\rho$ without a linear term}'', since there always exists a linear combination of the densities which does not have a linear term.
\end{Rem}

For a given equation (\ref{Fsymb}), to find the vector space $V_F$ of its admissible linear terms is a non-trivial problem. In the differential case this problem was completely solved for scalar polynomial homogeneous evolutionary partial differential equations \cite{mr99g:35058}, for systems of two-component equations \cite{gz}, as well as for odd order non-evolutionary equations \cite{nw07}. 
In next session we are going to formulate the necessary integrability conditions in the universal form independent on the structure of the vector space of its admissible linear terms.

\section{Integrability conditions for differential-difference equations}\label{sec5}

In the case of partial differential equations universal integrability conditions can be formulated in terms of a formal recursion operator \cite{mr86i:58070}-\cite{mr93b:58070}. Namely, the  existence of an infinite hierarchy of symmetries implies the existence of a first order formal pseudo-differential series with the coefficients in the corresponding differential field, satisfying the same equations as the recursion operator. Universality means that this fact does not depend on unknown a priory possible gaps in the hierarchy of symmetries. Later on this theory has been reformulated in the symbolic representation \cite{mn1} which enable us to tackle some integro-differential and non-evolutionary equations \cite{mn2, mnw07, MNW3}. Universality follows from the existence of fractional powers of formal pseudo-differential series. In the differential-difference case a fractional power represented by a difference formal series with coefficients in the difference field $\cF$ or ring $\bar{\cR}$ may not exist \cite{MWX}. To tackle the problem we introduce in this section a {\em quasi-local} extension of the difference ring $\bar{\cR}$. It will enable us to formulate universal integrability conditions for differential-difference equations in the symbolic representation.

\subsection{Quasi-local extension  of the difference ring $\hat{\cR}$. }\label{sec51}

In Section \ref{doisr} we have shown that  formal difference series with coefficients from $\bar{\cR}$ in the symbolic representation take the form
\begin{equation}\label{sAa}
 \hat{A}=\sum\limits_{n=0}^{\infty} \hat{u}^n A_n(\xi_1,\ldots,\xi_n,\eta)
,\qquad 
 A_n(\xi_1,\ldots,\xi_n,\eta)\in\Xi_n[\eta^{\pm 1}],
\end{equation}
where $\Xi_n[ \eta^{\pm 1}]=\k[\xi_1^{\pm 1},\ldots,\xi_n^{\pm 1}]^{S_n}[\eta^{\pm 1}]$ is a set of Laurent polynomials in the variable $\eta$, whose coefficients are symmetric Laurent polynomials in  variables $\xi_1,\ldots,\xi_n$.
The composition law for difference operators in the symbolic representation is given by (\ref{comprule}). 

Here we define a set of formal series 
\begin{equation}\label{aaa}
 \fA=\{\sum\limits_{n=0}^{\infty} \hat{u}^n  A_n(\xi_1,\ldots,\xi_n,\eta)\, |\,
  A_n(\xi_1,\ldots,\xi_n,\eta)\in\k(\xi_1,\ldots,\xi_n,\eta)^{S_n}\},
\end{equation}
where $ A_n(\xi_1,\ldots,\xi_n,\eta)$ are {\em rational functions} in its variables, symmetric with respect to permutations of the $\xi$--variables. The natural addition and the composition rule (\ref{comprule}) define on $\fA$ a structure of a non-commutative ring. Obviously, difference operators in the symbolic representation and formal difference series, such as (\ref{sAa}) belong to $\fA$, but in general elements of $\fA$ do not represent formal difference series or difference operators.

\begin{Def}\label{defloc} Let $A_n= A_n(\xi_1,\ldots,\xi_n,\eta)$ be a rational function of its variables.
The term $\hu^n A_n$  is called 
 {\em L--local (M--local)}, if the coefficients $ A_{n,k}(\xi_1,\ldots,\xi_n)$ of its power expansion in the variable $\eta$ at infinity 
 $
   A_{n}=\sum_{k\leqslant p_n}A_{n,k}(\xi_1,\ldots,\xi_n)\eta^k
$ (resp.  at zero $
   A_{n}=\sum_{k \geqslant q_n}A_{n,k}(\xi_1,\ldots,\xi_n)\eta^k
$)
are symmetric Laurent polynomials in the variables $\xi_1,\ldots,\xi_n$. \\
A formal series 
 \begin{equation}
  \label{A}
A=\sum\limits_{n=0}^\infty \hu^n A_n(\xi_1,\ldots,\xi_n,\eta)
 \end{equation} 
 is called {\em L--local (M--local)} if all its terms are L--local (resp. M--local). The formal series $A$ (\ref{A}) is called {\em local} if it is L and M local.
\end{Def}
\begin{Ex} The following terms are local, they both M--local and  L--local:
\begin{equation*}
 \hu A_1(\xi_1,\eta)=\hu\frac{\eta (\eta+\xi_1)}{\eta-1}, \quad 
\hu^2 A_2(\xi_1,\xi_2,\eta)=\hu^2\frac{\eta(\eta \xi_1 \xi_2^2+\eta \xi_1^2 \xi_2-\eta \xi_1^2-\eta \xi_2^2-\eta \xi_1-\eta \xi_2-2 \xi_1 \xi_2)}
 {2(\eta-1)(\eta \xi_1-1)(\eta \xi_2-1)}.
\end{equation*}
They correspond to the first two terms of the canonical formal recursion operator (we define it in Section \ref{sec52}) $\Lambda=\eta+\hu A_1(\xi_1,\eta)+\hu^2 A_2(\xi_1,\xi_2,\eta)+\cdots$  for the Volterra chain $u_t=u_1-u_{-1}+u(u_1-u_{-1})$.
 
The term $\hu^2 (\eta+\xi_1+\xi_2)^{-1}$ is L--local, but not M--local. 
\end{Ex}
%Local formal series form a subring $\fA_0\subset\fA$. 

Let $\hat{A}\in\fA$ be a formal series of the form
\begin{equation}
\label{Aformal}
A=\eta^N+\sum_{p\ge 1}\hu^pa_p(\xi_1,\ldots,\xi_p,\eta),
\end{equation}
where $N$ is a positive integer, and let us formally seek its $N$--th root
\begin{equation}
\label{Bformal}
B=\eta+\sum_{p\ge 1}\hu^pb_p(\xi_1,\ldots,\xi_p,\eta)
\end{equation}
such that
%\begin{equation}
%\label{Bnformal}
$B^N=A$.
%\end{equation}
Using the composition rule (\ref{comprule}) and taking projections on the homogeneous components 
 $\hat{\pi}_s(B^N-A)=0,\,\,s=0,1,2, \cdots$, we obtain:
\begin{eqnarray}
\label{b1}
\hat{\pi}_1:&& \eta^{N-1}\Theta_N(\xi_1)b_1(\xi_1,\eta)-a_1(\xi_1,\eta)=0,
\\ \label{bs}
\hat{\pi}_s:&& \eta^{N-1}\Theta_N(\xi_1\xi_2\cdots \xi_s
)b_s(\xi_1,\ldots,\xi_s,\eta)%\\ \nonumber
+
f_s-a_s(\xi_1,\ldots,\xi_s,\eta)=0,
\end{eqnarray}
where function $\Theta_N$ is defined as
\begin{equation}\label{Theta}
 \Theta_N(\xi)\defeq 1+\xi+\cdots +\xi^{N-1}=(1-\xi^N)(1-\xi)^{-1},
\end{equation}
and the functions
$f_s$ depend only on $b_1,\ldots,b_{s-1}$ and $\eta$. For example 
\[
f_2=\sum_{n=0}^{N-2}\left\langle\ \sum_{m=0}^{N-n-2}\xi_1^n\xi_2^{n+m}b_1(\xi_1,\eta\xi_2)b_1(\xi_2,\eta)\right\rangle_{\Sigma_2}\eta^{N-2}.
\]
Relations (\ref{b1}), (\ref{bs}) form a triangular system of equations which enable us to find the rational functions $b_1(\xi_1,\eta),b_2(\xi_1,\xi_2,\eta),\ldots$ successively
\[
 b_1(\xi_1,\eta)=\frac{a_1(\xi_1,\eta)\eta^{1-N}}{\Theta_N(\xi_1)},\quad  b_2(\xi_1,\xi_2,\eta)=\frac{a_2(\xi_1,\xi_2,\eta)-f_2}{\eta^{N-1}\Theta_N(\xi_1\xi_2)},\ldots \ .
\]
For every $A$ of the form (\ref{Aformal}) we can find a unique formal series (\ref{Bformal}) satisfying the equation $B^N=A$. 

Let $A$ be a local series. Then from the relations (\ref{b1}), (\ref{bs}) it follows that the elements of $B$ are generally no longer local. If $\Theta_N(\xi_1)$ does not divide $a_1(\xi_1,\eta)$ then the coefficients of the power expansion of $ b_1(\xi_1,\eta)$ in $\eta$ as $\eta\to0$ and in $\eta^{-1}$ as $\eta\to\infty$ contain $\Theta_N(\xi_1)$ in their denominators and thus fail to be Laurent polynomials in $\xi_1$ and therefore they do not represent symbols of difference polynomials. Similarly, the coefficients of the expansions of $b_2(\xi_1,\xi_2,\eta)$ may contain $\Theta_N(\xi_1),\ \Theta_N(\xi_2)$ and $ \Theta_N(\xi_1\xi_2)$ in their denominators, etc. It motivates us to define  {\em $\Theta_N$-quasi-local} extension of the difference ring $\hat{\cR}$.

The action of the pseudo-difference operator  $\theta_N =\Theta_N^{-1}(\eta)$ on  $\hu^k a_k(\xi_1,\ldots,\xi_k)\in\hat{\cR}$ is given by
\[
 \theta_N (\eta)(\hu^k a_k(\xi_1,\ldots,\xi_k))=\hu^k \frac{a_k(\xi_1,\ldots,\xi_k)}{\Theta_N(\prod_{i=1}^k \xi_i)}.
\]
Let us define the sequence of the ring extensions
\[
  \check{\cR}_{(0)}=\hat{\cR},\quad  \check{\cR}_{(s+1)}=\overline{\check{\ring}_{(s)}\bigcup\theta_N (\check{\ring}_{(s)})},\qquad  s=0,1,2\ldots\ .
\]
Here the horizontal line denotes the ring closure as Abelian groups and with respect to the $\star$ product (\ref{starprod}). The index $s$ in $\check{\ring}_{(s)}$ shows the maximal ``nesting'' degree of $\theta_N $. Obviously 
\[
 \check{\cR}_{(0)}=\hat{\cR}\subset\check{\cR}_{(1)}\subset\check{\cR}_{(2)}\subset \cdots\ .
\]
The  quasi-local extension of the ring $\hat{\cR}$ is defined as the limit 
\begin{equation}\label{chekr}
 \check{\cR}=\lim\limits_{s\to\infty}\check{\cR}_{(s)}\, .
\end{equation}
Saying that $a\in\check{\cR}$ we mean that there exists such non-negative integer $s$ that $a\in\check{\cR}_{(s)}$. Elements of $\check{\cR}_{(s)},\ s\geqslant 1$ are called  quasi-local. The  quasi-local extension defined above depends on the choice of the integer $N$. 
\iffalse
We could define a universal quasi-local extension $\mathring{\cR}$ resulting from the sequence
\[
  \mathring{\cR}_{(0)}=\hat{\cR},\quad  \mathring{\cR}_{(s+1)}=\overline{\mathring{\ring}_{(s)}\bigcup\limits_{N=2}^\infty \theta_N (\mathring{\ring}_{(s)})},\qquad  s=0,1,2\ldots\ ,
\]
and $\mathring{\cR}=\lim\limits_{s\to\infty}\mathring{\cR}_{(s)}$, but $\check{\cR}$  extension is sufficient for our purposes.
%We now extend the definition of locality of formal series as follows:
\fi
\begin{Def} \label{defqlocal}
An element $\hu^pa(\xi_1,\ldots,\xi_p,\eta)\in\fA$ is called  $\ell$-quasi-local if the first $\ell$ terms $\hu^p a_i(\xi_1,\ldots,\xi_p)$ of its power expansion in $\eta^{-1}$ at $\eta\to\infty$
$$
\hu^p a(\xi_1,\ldots,\xi_p,\eta)=\hu^p a_1(\xi_1,\ldots,\xi_p)\eta^q+\hu^p a_2(\xi_1,\ldots,\xi_p)\eta^{q-1}+\cdots
$$
are quasi-local. 
 
A formal series $A=\phi(\eta)+\sum_{p\ge 1}\hu^pa_p(\xi_1,\ldots,\xi_p,\eta),\,\,\phi(\eta)\in\k[\eta,\eta^{-1}]$ is called $\ell$-quasi-local if all its terms are $\ell$-quasi-local. It is call quasi-local, if it is $\ell$-quasi-local for all $\ell$.
\end{Def}
Above we defined the Laurent  quasi-locality considering power series  expansions in $\eta^{-1}$.
Similarly we could define the Maclaurin  quasi-local formal series using the power expansion in $\eta$ (cf. Definition \ref{defloc}). In the next Section \ref{sec52}, as well as in applications to the problem of classification of integrable equations of the form (\ref{gen}) in Section \ref{sec61} it will be sufficient to use concepts of Laurent  quasi-locality only (which will be addressed as quasi-locality for shortness if $N$ is defined).  

The above computation of the $N$-th root of a formal series $A$ (\ref{Aformal}) can be recast in the following Proposition.
\begin{Pro}\label{Nroot}
Let 
$$
A=\eta^N+\sum_{p\ge 1}\hu^pa_p(\xi_1,\ldots,\xi_p,\eta)
$$
be a formal series whose first $k$ terms $\hu a_1(\xi_1,\eta),\ldots,\hu^k a_k(\xi_1,\ldots,\xi_k,\eta)$ are $\ell$-quasi-local.  Then there exists a unique formal series
$$B=\eta+\sum_{p\ge 1}\hu^pb_p(\xi_1,\ldots,\xi_p,\eta)$$ satisfying the equation $B^N=A$, and first $k$ terms of series $B$ are $\ell$-quasi-local.
\end{Pro}

The above proposition admits an immediate generalisation:
\begin{Pro}
\label{proP}
Let $P(\eta)=\sum\limits_{k\leqslant N} c_k \eta^k,\ c_N\ne 0,\ N\ne 0$  and
$$
A=P(\eta)+\sum_{p\ge 1}\hu^pa_p(\xi_1,\ldots,\xi_p,\eta)
$$
be a formal series whose first $k$ terms $\hu a_1(\xi_1,\eta),\ldots,\hu^k a_k(\xi_1,\ldots,\xi_k,\eta)$ are $\ell$-quasi-local. Then there exists a unique  formal series
$$B=\eta+\sum_{p\ge 1}\hu^pb_p(\xi_1,\ldots,\xi_p,\eta)$$ satisfying the equation
$
P(B)=A,
$
and first $k$ terms $\hu b_1(\xi_1,\eta),\ldots,\hu^k b_k(\xi_1,\ldots,\xi_k,\eta)$ are $\ell$-quasi-local.
\end{Pro}
Clearly, if formal series $A$ in Propositions \ref{Nroot} and \ref{proP} is local or quasi-local, then the series $B$ is quasi-local. 
\iffalse
{\bf Proof}: We equate projections $\hat{\pi}_s(P(B)-A)=0,\,\,s=1,2,\ldots$:
\begin{eqnarray*}
\hat{\pi}_1(P(B)-A)&=&0:\,\,\bigg[\sum_{k=1}^nc_k\eta^{k-1}(1+\xi_1+\cdots+\xi_1^{k-1})-\\&&-\xi_1^{-1}\sum_{k=1}^nc_{-k}\eta^{-k-1}(1+\xi_1^{-1}+\cdots+\xi_1^{-k+1})\bigg]b_1(\xi_1,\eta)=a_1(\xi_1,\eta),
\\
\hat{\pi}_s(P(B)-A)&=&0,\,\, s=2,3,\ldots,:\\&&
\bigg[\sum_{k=1}^nc_k\eta^{k-1}(1+(\xi_1\cdots\xi_s)+\cdots+(\xi_1\cdots\xi_s)^{k-1})-
\\
&&-(\xi_1\cdots\xi_s)^{-1}\sum_{k=1}^nc_{-k}\eta^{-k-1}(1+(\xi_1\cdots\xi_s)^{-1}+\cdots+(\xi_1\cdots\xi_s)^{-k+1})) \bigg]b_s(\xi_1,\ldots,\xi_s,\eta)=
\\
&&=a_s(\xi_1,\ldots,\xi_s,\eta)+f_s(b_1(\xi_1,\eta),b_2(\xi_1,\xi_2,\eta),\ldots,b_{s-1}(\xi_1,\ldots,\xi_{s-1},\eta)),
\end{eqnarray*}
where $f_s(b_1(\xi_1,\eta),b_2(\xi_1,\xi_2,\eta),\ldots,b_{s-1}(\xi_1,\ldots,\xi_{s-1},\eta))$ are  polynomials in lower order terms $b_1,\ldots,b_{s-1}$. These relations form a triangular linear system of equations on $b_1(\xi_1,\eta),\ldots,b_s(\xi_1,\ldots,\xi_s,\eta),\ldots$. Therefore, there exists a unique quasi-local formal series $B$ such that $P(B)=A$. $\square$
\fi

\subsection{Canonical formal recursion operator}\label{sec52}

We now proceed to the construction of the universal form of necessary integrability conditions for finite order evolutionary differential-difference equations, i.e., equations whose right hand side is either a polynomial or a formal series in a finite number of variables $u_{-n},\ldots, u_n$:
\begin{equation}\label{eq50}
 u_t=f=f(u_{-n},\ldots, u_n),\qquad f=\sum_{i\ge 1}f^{(i)},\qquad \ \ f^{(1)}\ne 0,\ \ f^{(i)}\in\ring^i \cap \k[u_{-n},\ldots, u_n].\ 
\end{equation}
In symbolic representation equation (\ref{eq50}) takes the form 
\begin{equation}\label{eqsymb}
 \hu_t=\hat{f}=\hu\omega(\xi_1)+\sum_{i\ge 2}\hu^i a_i(\xi_1,\ldots,\xi_i),\qquad a_i\in\Xi_i.
\end{equation}
We shall assume that  $\omega(\xi_1)\ne const$ (see Remark \ref{omegaconst}).
\begin{Def} A quasi-local formal series
$$
\hat{\Lambda}=\phi(\eta)+\sum_{p\ge 1}\hu^p\phi_p(\xi_1,\ldots,\xi_p,\eta),\quad\phi(\eta)\in\k(\eta),
$$
is called a formal recursion operator for  equation (\ref{eqsymb})
if $\hat{\Lambda}$ satisfies the equation
\begin{equation}\label{forsym}
\hat{\Lambda}_t-\hat{f}_*\circ\hat{\Lambda}+\hat{\Lambda}\circ\hat{f}_*=0\,.
\end{equation}
\end{Def}

Equation (\ref{forsym}) is linear in $\hat{\Lambda}$ and for any rational function $\phi(\eta)$ it has a unique solution in terms of a formal series with rational coefficients. Namely, the following theorem holds:

\begin{Thm}\label{theorLam} Let
\begin{equation}\label{Lambrat}
\Lambda=\phi(\eta)+\sum_{p\ge 1}\hu^p\phi_p(\xi_1,\ldots,\xi_p,\eta),\quad \phi(\eta)\in\k(\eta),\ 
\phi_p(\xi_1,\ldots,\xi_p,\eta)\in\k(\xi_1,\ldots,\xi_p,\eta)
\end{equation}
be a formal series with rational coefficients satisfying equation (\ref{forsym}).
Then, for any choice of a rational function $\phi(\eta)$, the coefficients of the series can be recursively determined from the following system: 
\begin{eqnarray}
\label{phi1}
\phi_1(\xi_1,\eta)&=&2\frac{(\phi(\xi_1\eta)-\phi(\eta))}{G^{\omega}(\xi_1,\eta)}a_2(\xi_1,\eta),\\ 
\label{phip}
\phi_p(\xi_1,\ldots,\xi_p,\eta)&=&(p+1)\frac{(\phi(\xi_1\cdots\xi_p\eta)-\phi(\eta))a_{p+1}(\xi_1,\ldots,\xi_p,\eta)}{G^{\omega}(\xi_1,\ldots,\xi_p,\eta)}+\frac{1}{G^{\omega}(\xi_1,\ldots,\xi_p,\eta)}\times\\ \nonumber
&&\times\bigg[\sum_{s=1}^{p-1}s\langle\phi_s(\xi_1,\ldots,\xi_{s-1},\xi_s\cdots\xi_p,\eta)a_{p-s+1}(\xi_s,\ldots,\xi_p)\rangle+\\ \nonumber
&&+\sum_{s=2}^ps\big(\langle\phi_{p-s+1}(\xi_1,\ldots,\xi_{p-s+1},\xi_{p-s+2}\cdots\xi_p\eta)a_s(\xi_{p-s+2},\ldots,\xi_p,\eta)\rangle-\\ \nonumber
&&-\langle a_s(\xi_1,\ldots,\xi_{s-1},\xi_s\cdots\xi_p\eta)\phi_{p-s+1}(\xi_s,\ldots,\xi_p,\eta)\rangle\big)\bigg],\quad p=2,3\ldots.
\end{eqnarray}
\end{Thm}
\begin{proof}
The proof follows immediately from the observation that equations  
$$\hat{\pi}_p(D_t(\Lambda)-\hat{f}_*\circ\Lambda+\Lambda\circ\hat{f}_*)=0,\,\,p=1,2,\ldots$$ 
can be solved recursively with respect to $\phi_p(\xi_1,\ldots,\xi_p,\eta),\,\,p=1,2,\ldots$. 
\end{proof}

Theorem \ref{theorLam} does not mean that any equation (\ref{eqsymb}) possesses a formal recursion operator, since it does not guarantee that the formal series obtained is quasi-local. Below we are going to show that for  integrable equations, i.e., equations possessing an infinite hierarchy of symmetries, the obtained series must be quasi-local. Therefore conditions of quasi-locality of the terms $\hu^p\phi_p(\xi_1,\ldots,\xi_p,\eta), p=1,2,\ldots $ obtained in Theorem \ref{theorLam} are necessary  integrability conditions for equation (\ref{eqsymb}).

Let $\hat{\Lambda}$ be a formal recursion operator, then powers of $\hat{\Lambda}$ and  linear combinations of powers with constant coefficients are also formal recursion operators. Indeed, they satisfy equation (\ref{forsym}) and are quasi-local. A constant is a (trivial) recursion operator. Moreover, for any formal recursion operator $\hat{\Lambda}$ with non-constant $\phi(\eta)$ there exist a unique quasi-local formal series  
\begin{equation}\label{cansym}
\Lambda=\eta+\sum_{p\ge 1}\hu^p\phi_p(\xi_1,\ldots,\xi_p,\eta)
\end{equation}
satisfying the equation $\hat{\Lambda}=\phi(\Lambda)$ (Proposition \ref{proP}). It is easy to see that $\Lambda$ also satisfies equation (\ref{forsym}). A formal recursion operator with $\phi(\eta)=\eta$ we say it is  {\it canonical}.

The following Theorem shows that the existence of a formal recursion operator follows from integrability of equation (\ref{eqsymb}).

\begin{Thm}\label{main} Let the algebra of symmetries of equation (\ref{eqsymb}) be infinite dimensional. Then a formal series (\ref{Lambrat}) satisfying equation (\ref{forsym}) is quasi-local. 
\end{Thm}
\begin{proof}
Let $G_{k},\ k=1,2,\ldots$ 
be an infinite sequence of symmetries. It follows from Proposition \ref{prolin} that  symmetries have non-vanishing linear part
$\hu \Omega_k(\xi_1)=\hat{\pi}_1 (G_k)$ and the linear space of the Laurent polynomials $\{\Omega_{k }(\xi)=\sum_{m=M_k}^{N_k}c_{k,m}\xi^m\, |\, c_{k,m}\in\k,\ M_k,N_k\in\Z\}$ is infinite dimensional. Let  the order of the symmetry $G_k$  be defined as $\deg^+_\xi(\Omega_k(\xi ))\defeq N_k$. Without loss of generality we can assume that the sequence $N_1<N_2<N_3<\cdots$ is unbounded (otherwise we can apply the reflection operator $\cT$ to the equation and its symmetries). Let $\hat{G}$ be a symmetry of order $N$
$$
\hat{G}=\hu\Omega(\xi_1)+\sum_{p\ge 2}\hu^pA_{p}(\xi_1,\ldots,\xi_p),\qquad \deg^+_\xi(\Omega(\xi))= N.
$$
By Definition \ref{defsym} we have $[G,f]=0$ and therefore
\[
 ([G,f])_*=(G_*)_t+G_*\circ f_*-(f_*)_\tau-f_*\circ G_*=0,
\]
where $\partial_\tau=X_G$ is the derivation defined by the evolutionary equation $u_\tau=G$. In symbolic representation it reads
\begin{equation}
\label{thmeq1}
(\hat{G}_{*})_t-\hat{f}_*\circ \hat{G}_{*}+\hat{G}_{*}\circ \hat{f}_*=(\hat{f}_*)_\tau
\end{equation}
where
$$
\hat{G}_{*}=\Omega(\eta)+\sum_{p\ge 2}p\hu^{p-1}A_{p}(\xi_1,\ldots,\xi_{p-1},\eta)
$$
is a local formal series. Its  coefficients are symmetric Laurent polynomials in all variables and $\deg^+_\eta(A_{p}(\xi_1,\ldots,\xi_{p-1},\eta))\leqslant N$ (Proposition \ref{profre}, Theorem \ref{theorsym}). Let us represent the coefficients of $\hat{f}_*$ and $\hat{G}_{*}$ as power series in $\eta^{-1}$
\begin{equation}
 \begin{array}{ll}
  \hat{f}_*=\sum\limits_{p\geqslant1} \, \sum\limits_{k \leqslant n}  pa_{p,k}(\xi_1,\ldots,\xi_{p-1})\eta^k\hu^{p-1}\qquad &  \omega(\xi_1)=\sum\limits_{k \leqslant n}  a_{1,k}(\xi_1)\eta^k,\\&\\
   \hat{G}_*=\sum\limits_{p\geqslant1}\, \sum\limits_{k \leqslant N}  pA_{p,k}(\xi_1,\ldots,\xi_{p-1})\eta^k\hu^{p-1}\qquad &  \Omega(\xi_1)=\sum\limits_{k \leqslant N} A_{1,k}(\xi_1)\eta^k
 \end{array}
\end{equation}
and substitute these expansions in equation (\ref{thmeq1}). We notice that the  degree of $\eta$ in the right hand side of the equation is $\deg^+_\eta(\hat{f}_*)_\tau\leqslant n$ but the degree of the left hand side is $N+n$. Thus, at least first $N$ terms  $pA_{p,k}(\xi_1,\ldots,\xi_{p-1})\eta^k\hu^{p-1},\ k=N,N-1,\ldots,1,\ p\in\N$ of the expansion $\hat{G}_{*}$ satisfy the homogeneous linear equation  (\ref{forsym}) with $\phi(\xi)=\Omega(\xi)$.  The solution of the equation exists and is unique (Theorem \ref{theorLam}). Thus we can identify 
these terms with $\phi_{p-1,k}(\xi_1,\ldots,\xi_{p-1})\eta^k \hu^{p-1}$ in the expansion of 
$\phi_{p-1}(\xi_1,\ldots,\xi_{p-1},\eta)$. Since the Fr\'echet derivative $\hat{G}_*$ of a symmetry $G$ is local, the obtained solution $\Lambda$ is $N$-local. It follows from Proposition \ref{proP} that there exist $N$-quasi-local (canonical) series of the form (\ref{cansym}). Moreover this series in quasi-local since $N$ can be taken arbitrary large. 
\end{proof}

Theorem \ref{main} is constructive, and it provides necessary integrability conditions for equation (\ref{eqsymb}), independent on the symmetry structure of the equation. The fact of existence of a formal recursion operator can also be proved using Adler's theorem \cite{adler14}. It follows from Adler's Theorem as well as from Theorem \ref{theorLam} that for equations of order $N$ the coefficients of the formal recursion operator $\Lambda$ (\ref{Lambrat}) belong to $\Theta_N$-quasi-local extension. Adler's Theorem also suggests that there exists a rational function $\phi(\xi_1)$ such that $\Lambda$ (\ref{Lambrat}) is local. We conjecture   that $\phi(\xi_1)=\omega(\xi_1)_+$, that is, the polynomial part of the Laurent polynomial $\omega(\xi_1)$, results in a local $\Lambda$. We have verified it for a number of equations.

Theorem \ref{theorLam}  shows the advantage of symbolic representation. It provides explicit recurrence formulae  (\ref{phi1}) and (\ref{phip}) for the  coefficients of a formal recursion operator.
We can use it to tackle the classification problem of integrable differential-difference equations. 
For a given family of equations of form (\ref{eqsymb}) the process is as follows:
\begin{itemize}
\item Use (\ref{phi1}) and (\ref{phip}) to find a few first coefficients  $\phi_p(\xi_1,\ldots,\xi_p,\eta)$.
\item Find constraints on the equations imposed by the  quasi-locality conditions of $\hu^p \phi_p(\xi_1,\ldots,\xi_p,\eta)$.
\end{itemize}
We illustrate the procedure in the following simple examples.

\begin{Ex} Let us describe all integrable equations of the form
\begin{equation}
\label{toy1} u_t=u_2+\alpha u+u(u_2+\beta u_1+\gamma u),\quad\alpha,\beta,\gamma\in\C.
\end{equation}
Its symbolic representation is of the form (\ref{eqsymb}) with
\[
\omega(\xi_1)=\xi_1^2+\alpha,\quad a_2(\xi_1,\xi_2)=\frac{1}{2}(\xi_1^2+\xi_2^2)+\frac{\beta}{2}(\xi_1+\xi_2)+\gamma
\]
and $a_s(\xi_1,\ldots,\xi_s)=0,\,\, s>2$. Using Theorem \ref{theorLam} we recursively compute the coefficients $\phi_1,\phi_2$ of the formal recursion operator starting with $\phi(\eta)=\eta$. The first coefficient reads as
\[
\phi_1(\xi_1,\eta)=\frac{\eta  \left(\xi _1-1\right) \left(\eta ^2+\beta  \eta +\xi _1^2+\beta  \xi _1+2 \gamma \right)}{\eta ^2 (\xi _1^2-1)-\xi _1^2-\alpha },
\]
and its power expansion in $\eta$ at $\eta\to\infty$ is of the form
\[
\phi_1(\xi_1,\eta)=\frac{\eta+\beta}{\xi _1+1}+\frac{\xi _1^4+\beta  \xi _1^3+2 \gamma  \xi _1^2-\beta  \xi _1+\alpha -2 \gamma }{\left(\xi _1-1\right) \left(\xi _1+1\right){}^2}\eta^{-1}+O(\eta^{-2}).
\]
We see that the element $\hu\phi_1(\xi_1,\eta)$ is $2$-quasi-local as
\[
\hu \frac{\eta+\beta}{\xi _1+1}\in \check{\cR}_{(1)},
\] 
and the ring extension is performed by using the operator $\theta_2=(1+\eta)^{-1}$.
The third term in the expansion contains $\xi_1-1$ in the denominator, and therefore 
\[
\hu \frac{\xi _1^4+\beta  \xi _1^3+2 \gamma  \xi _1^2-\beta  \xi _1+\alpha -2 \gamma }{\left(\xi _1-1\right) \left(\xi _1+1\right){}^2}\eta^{-1}
\]
is quasi-local if and only if $\xi_1-1$ divides $\xi _1^4+\beta  \xi _1^3+2 \gamma  \xi _1^2-\beta  \xi _1+\alpha -2 \gamma $. The latter occurs if and only if $\alpha=-1$. Then, if $\alpha=-1$, the first term of the formal recursion operator reads as
\[
\hu\phi_1(\xi_1,\eta)=\hu \frac{\eta  \left(\eta ^2+\beta  \eta+\xi _1^2 +\beta  \xi _1+2 \gamma \right)}{(\eta -1) (\eta +1) \left(\xi _1+1\right)},
\]
and it is easy to see that this term is $\Theta_2$-quasi-local.

Let $\alpha=-1$. We now consider the power expansion in $\eta$ at $\eta\to\infty$ of the second coefficient $\phi_2(\xi_1,\xi_2,\eta)$:
\[
\phi_2(\xi_1,\xi_2,\eta)=-\frac{\xi _1+\xi _2}{2 \left(\xi _1+1\right) \left(\xi _2+1\right) \left(\xi _1 \xi _2+1\right)}\eta-\frac{\beta  \left(\xi _1+\xi _2+2\right)}{2 \left(\xi _1+1\right) \left(\xi _2+1\right) \left(\xi _1 \xi _2+1\right)}
\]
\[
+\frac{P(\xi_1,\xi_2,\beta,\gamma)}{2\left(\xi _1 \xi _2-1\right) \xi _1\xi _2 \left(\xi _1+1\right)  \left(\xi _2+1\right)  \left(\xi _1 \xi _2+1\right){}^2}+O(\eta^{-2}),
\]
where $P(\xi_1,\xi_2,\beta,\gamma)$ is a polynomial in its variables. The third term in this expansion is the obstruction to quasi-locality of $\hu^2\phi_2(\xi_1,\xi_2,\eta)$, unless $\xi_1\xi_2-1$ divides $P(\xi_1,\xi_2,\beta,\gamma)$. We have
\[
P(\xi_1,\xi_1^{-1},\beta,\gamma)=-\left(1+\xi_1^{-1}\right){}^2 \left(2 (\gamma+1)  \xi _1+\beta(1+\xi _1^2) \right).
\] 
So the division occurs if and only if $\beta=0$ and $\gamma=-1$. The latter implies $\Theta_2$-quasi-locality of the element $\hu^2\phi_2(\xi_1,\xi_2,\eta)$. The resulting equation (\ref{toy1}) is
\[
u_t=u_2-u+u(u_2-u)=(u+1)(u_2-u).
\]
Upon the change of variables $u_k\to u_k-1,\,\,k\in\Z$, it becomes the stretched Burgers equation 
\[
u_t=u(u_2-u).
\]
\end{Ex}

\begin{Ex}
 Let us find all integrable equations of the form
\begin{equation}
\label{toy}
u_t=u_2+\alpha u_1-\alpha u_{-1}-u_{-2}+u(u_2+\beta u_1-\beta u_{-1}-u_{-2}),\quad \alpha,\beta\in\C.
\end{equation}
Its symbolic representation is of the form (\ref{eqsymb}), where
\[
\omega(\xi)=\xi^2+\alpha\xi-\alpha\xi^{-1}-\xi^{-2},\quad a_2(\xi_1,\xi_2)=\frac{1}{2}(\xi_1^2+\xi_2^2-\xi_{1}^{-2}-\xi_{2}^{-2}+\beta(\xi_1+\xi_2-\xi_1^{-1}-\xi_2^{-1}))
\]
and $a_s(\xi_1,\ldots,\xi_s)=0,\,\, s>2$. Using the theorem \ref{theorLam} we recursively compute the coefficients $\phi_1,\phi_2,\phi_3$ of the formal recursion operator starting with $\phi(\eta)=\eta$. We have
\[
\phi_1(\xi_1,\eta)=\frac{\eta  \left(\beta  \eta ^2 \xi _1+\beta  \eta  \xi _1^2+\eta ^3 \xi _1+\eta ^2+\eta  \xi _1^3+\xi _1^2\right)}{(\eta -1) \left(\alpha  \eta  \xi _1+\eta ^2 \xi _1^2+\eta ^2 \xi _1+\eta  \xi _1^2+2 \eta  \xi _1+\eta +\xi _1+1\right)},
\]
and it is easy to see that this term is quasi-local as all coefficients of its power expansion at $\eta\to\infty$ (as well as at $\eta\to 0$) are $\Theta_2$-quasi-local. 

The direct computation shows that the next coefficient $\phi_2(\xi_1,\xi_2,\eta)$ is of the form
\[
\phi_2=\frac{\Phi_2(\xi_1,\xi_2,\eta,\alpha,\beta)}{\left(\eta  \xi _1-1\right) \left(\eta  \xi _2-1\right) \left(\alpha  \eta  \xi _1 \xi _2+\eta ^2 \xi _1^2 \xi _2^2+\eta ^2 \xi _1 \xi _2+\eta  \xi _1 \xi _2^2+\eta  \xi _1+\eta  \xi _1^2 \xi _2+\eta  \xi _2+\xi _1 \xi _2+1\right)},
\]
where $\Phi_2(\xi_1,\xi_2,\eta,\alpha,\beta)$ is a polynomial in its arguments. From this it follows that the second term $\hu^2\phi_2(\xi_1,\xi_2,\eta)$ is quasi-local.

The quasi-locality condition of next coefficient $\phi_3(\xi_1,\xi_2,\xi_3,\eta)$ imposes restrictions on parameters $\alpha,\beta$. The function $\phi_3$ can be represented as
\[
\phi_3=\frac{\Phi_3(\xi_1,\xi_2,\xi_3,\eta,\alpha,\beta)}{\Psi_3(\xi_1,\xi_2,\xi_3,\eta,\alpha)},
\]
where $\Phi_3,\Psi_3$ are polynomials in their arguments. The polynomial $\Psi_3$ contains the irreducible factor
\begin{eqnarray*}
(\xi_1\xi_2\xi_3\eta)^2G^{\omega}(\xi_1,\xi_2,\xi_3,\eta)&=&\left(\xi _1^2 \xi _2^2 \xi _3^2-1\right) \left(\eta ^4 \xi _1^2 \xi _2^2 \xi _3^2+1\right)+\alpha  \eta  \xi _1 \xi _2 \xi _3 \left(\xi _1 \xi _2 \xi _3-1\right) \left(\eta ^2 \xi _1 \xi _2 \xi _3+1\right)\\
&+&\eta ^2 (-\alpha  \xi _2^2 \xi _3^2 \xi _1^3-\alpha  \xi _2^2 \xi _3^3 \xi _1^2-\alpha  \xi _2^3 \xi _3^2 \xi _1^2+\alpha  \xi _2 \xi _3^2 \xi _1^2+\alpha  \xi _2^2 \xi _3 \xi _1^2+\alpha  \xi _2^2 \xi _3^2 \xi _1
\\&-&\xi _2^2 \xi _3^2 \xi _1^4-\xi _2^2 \xi _3^4 \xi _1^2+\xi _2^2 \xi _1^2-\xi _2^4 \xi _3^2 \xi _1^2+\xi _3^2 \xi _1^2+\xi _2^2 \xi _3^2).
\end{eqnarray*}
The presence of this factor results in violation of the quasi-locality of the term $\hu^3\phi_3$, unless it cancels out by the numerator $\Phi_3$. The cancellation takes place if and only if $\beta=\alpha$ and $\alpha=0$ or $\alpha=1$. If $\alpha=0$ then the resulting equation is
\[
u_t=u_2-u_{-2}+u(u_2-u_{-2})=(u+1)(u_2-u_{-2}),
\]
which upon the change of variables $u_k\to u_k-1,\,\,k\in\Z$, becomes the stretched Volterra equation 
\[
u_t=u(u_2-u_{-2}).
\]
In the case $\alpha=1$ we have
\[
u_t=u_2+u_1-u_{-1}-u_{-2}+u(u_2+u_1-u_{-1}-u_{-2})=(u+1)(u_2+u_1-u_{-1}-u_{-2}),
\]
and after the same change of variables $u_k\to u_k-1,\,\,k\in\Z$, we obtain the Narita-Itoh-Bogoyavlensky equation (\ref{NIB}) when $n=2$, that is,
\begin{equation}\label{NIB2}
u_t=u(u_2+u_1-u_{-1}-u_{-2}).
\end{equation}
\end{Ex}

\section{Classification of integrable differential-difference equations}\label{sec6}

Here we apply the previous section results to the problem of classification of anti-symmetric quasi-linear integrable  differential-difference equations of order $(-n,n)$  
\begin{equation}
\label{gen0}
u_{t}=u_n f(u_{n-1},\ldots,u_{1-n})-u_{-n}f(u_{1-n},\ldots,u_{n-1})+g(u_{n-1},\ldots,u_{1-n})-g(u_{1-n},\ldots,u_{n-1}), 
\end{equation}
\iffalse
\begin{equation}
\label{gen0}
u_{t}=u_n f+g-\cT(u_n f+g),\ \  f=f(u_{n-1},\ldots,u_{1-n}),\ g=g(u_{n-1},\ldots,u_{1-n}),\quad f(0,\ldots,0)=1,
\end{equation}
\fi
where $f,g$ are polynomial functions or formal series. Quasi-linear equations are called equivalent if they are related by  invertible transformations 
$u_k  \mapsto \alpha u_k+\beta,\ t\mapsto\gamma t,\ \alpha,\gamma\in\k^*,\ \beta\in\k$.  The equivalence classes may contain equations with $f(0,\ldots,0)= 0$ (see Example \ref{ex4}). In the classification list it is sufficient to present a single representative from each equivalence class. Integrable equations are members of infinite hierarchy of symmetries. Each hierarchy has a seed which is hierarchy member of a minimal  possible order. Thus instead of presenting all integrable equations of a certain fixed order, we only present the seeds of integrable hierarchies removing the ones possessing lower order symmetries.

%It make sense to classify integrable hierarchies or their seeds, rather than equations of a certain fixed order. 
%For example, integrable equation (\ref{Q2}) of order $(-3,3)$ is a member of the Volterra hierarchy.  

In this section we give a complete list for equations of the form \eqref{gen0} when $n=3$ satisfying necessary integrability conditions - the quasi-locality conditions for the canonical formal recursion operator. The integrability for each  equation from the list is proved by either using difference substitutions or presenting Lax representations. Moreover, we investigate the higher order integrable analogue for each equation from the list. Similar to the Narita-Itoh-Bogoyavlensky lattice (\ref{NIB}), this is a family of integrable equations for any specific $n\in \N$, and for any two different $n$, the corresponding flows do not commute
(cf. Section \ref{sec62}). Finally we present a Lax representation for a new integrable differential-difference hierarchy.

\subsection{Classification results}\label{sec61}

In this session we present the exhaustive list of integrable differential-difference equations of the form
\begin{equation}
\label{gen}
u_{t}=u_3f(u_2,u_1,u)-u_{-3}f(u_{-2},u_{-1},u)+g(u_2,u_1,u)-g(u_{-2},u_{-1},u),\qquad f(0,0,0)=1,
\end{equation}
where $f,g$ are polynomial functions or formal power series. We omit equations admitting  a symmetry of order $(1,-1)$, which have been studied in detail in \cite{yam83,Yami0}.
\begin{Thm}\label{class3} Every integrable differential-difference equation (\ref{gen}) with no symmetries of order $(-1,1)$  can be obtained from one of the equations in the following list
\begin{eqnarray}
\label{V1}
u_{t}&=&u(u_3-u_{-3}),\\
\label{V2}
u_{t}&=& u^2 (u_3-u_{-3}),\\
\label{V3}
u_{t}&=&(u^2+u)(u_3-u_{-3}),\\
\label{eq1}
u_{t}&=&u (u_1u_2u_3 - u u_1u_2 +u u_{-1}u_{-2}-u_{-1}u_{-2}u_{-3}),\\
\label{eq4}
u_{t}&=&u (u_2u_3-u_1u_2+u u_1- u u_{-1}+u_{-1}u_{-2}-u_{-2}u_{-3}),\\
\label{ser3}
u_{t}&=&u \left(\frac{u_3u_1}{u_2}-\frac{u_{-3}u_{-1}}{u_{-2}}\right)+u^2\left(\frac{u_2}{u_1}-\frac{u_{-2}}{u_{-1}}\right),\\
\label{ser2}
u_{t}&=&u \left(\frac{u_3}{u_2}-\frac{u_{-3}}{u_{-2}}\right)+u \left(\frac{u_2}{u_1}-\frac{u_{-2}}{u_{-1}}\right)+u_1-u_{-1},\\
\label{bg2}
u_{t}&=&u\left(u_3+u_2+u_1-u_{-1}-u_{-2}-u_{-3}\right),\\
\label{bg1}
u_{t}&=&u\left(u_1u_2u_3-u_{-1}u_{-2}u_{-3}\right),
\\
\label{bg11}
u_{t}&=&u^2\left(u_1u_2u_3-u_{-1}u_{-2}u_{-3}\right),
\\
\label{bg12}
u_{t}&=&(u^2+ u)\left(u_1u_2u_3-u_{-1}u_{-2}u_{-3}\right),
\\
\label{eq51}
u_{t}&=&u(u_1u_3+ u u_2-u u_{-2}-u_{-1}u_{-3}),\\
\label{eq31}
u_{t}&=&u (u_2u_3 +u u_1-u u_{-1}-u_{-2}u_{-3}),\\
\label{eq2}
u_{t}&=&u^2(u_1u_2u_3-u_{-1}u_{-2}u_{-3})- u(u_1u_2-u_{-1}u_{-2}),\\
\label{eq5}
u_{t}&=&u(u_1u_3+ u u_2-u u_{-2}-u_{-1}u_{-3})-u (u_2+u_1-u_{-1}-u_{-2}),\\
\label{eq3}
u_{t}&=&u (u_2u_3 +u u_1-u u_{-1}-u_{-2}u_{-3})- u(u_2+u_1-u_{-1}-u_{-2}),\\
\label{ser1}
u_{t}&=&(u^2+1)(u_3 \sqrt{u_1^2+1} \sqrt{u_2^2+1} - u_{-3} \sqrt{u_{-1}^2+1} \sqrt{u_{-2}^2+1}),
\end{eqnarray}
by shift $u_k\mapsto u_k+const$, re-scaling transformations $u_k\to \mu u_k,\,k\in\Z,\,\,t\to \nu t,\,\,\mu,\nu\in\C^*$ and, where necessary, a power expansion.
\end{Thm}
The proof of this classification theorem relies on the following result:
\begin{Pro} \label{uniq}
Let 
$$
u_{t}=\sum_{p\ge 1}f_p,\quad f_i\in\ring^i \quad \mbox{and} \quad u_{t}=\sum_{p\ge 1}\tilde{f}_p, \quad \tilde{f}_i\in\ring^i
$$
be two integrable formal differential-difference equations of form (\ref{gen}), such that
\begin{itemize}
\item The linear terms are $f_1=\tilde{f}_1=u_3+c_2u_2+c_1u_1-c_1u_{-1}-c_2u_{-2}-u_{-3}$;
\item The quadratic and cubic terms coincide: $f_2=\tilde{f}_2,\quad f_3=\tilde{f}_3$.
\end{itemize}
Then the formal differential-difference equations coincide, i.e. 
$
f_p=\tilde{f}_p,\quad \forall p\in\N.
$
\end{Pro}
\begin{proof}
Since the formal differential-difference equations are integrable, they possess the canonical formal recursion operators
$$
\Lambda=\eta+\sum_{p>0}\hu^p\phi_p(\xi_1,\ldots,\xi_p,\eta),\quad \tilde{\Lambda}=\eta+\sum_{p>0}\hu^p\tilde{\phi}_p(\xi_1,\ldots,\xi_p,\eta),
$$
such that all their terms are quasi-local. Moreover, since the linear, quadratic and cubic terms of the formal differential-difference equations coincide, the linear and quadratic terms of the formal recursion operators also coincide, i.e.
$$
\phi_1(\xi_1,\eta)=\tilde{\phi}_1(\xi_1,\eta),\quad \phi_2(\xi_1,\xi_2,\eta)=\tilde{\phi}_2(\xi_1,\xi_2,\eta).
$$
Consider cubic elements of $\Lambda$ and $\tilde{\Lambda}$: $\hu^3\phi_3(\xi_1,\xi_2,\xi_3,\eta)$ and $\hu^3\tilde{\phi}_3(\xi_1,\xi_2,\xi_3,\eta)$, and let $\hu^4a_4(\xi_1,\xi_2,\xi_3,\xi_4)$ and $\hu^4\tilde{a}_4(\xi_1,\xi_2,\xi_3,\xi_4)$ be the symbolic representations of $f_4$ and $\tilde{f}_4$. Then from (\ref{phip}) it follows that
\begin{equation}\label{ag4}
\phi_3(\xi_1,\xi_2,\xi_3,\eta)-\tilde{\phi}_3(\xi_1,\xi_2,\xi_3,\eta)=4\frac{\eta(\xi_1\xi_2\xi_3-1)(a_4(\xi_1,\xi_2,\xi_3,\eta)-\tilde{a}_4(\xi_1,\xi_2,\xi_3,\eta))}{\omega(\xi_1\xi_2\xi_3\eta)-\omega(\xi_1)-\omega(\xi_2)-\omega(\xi_3)-\omega(\eta)}
\end{equation}
since all the lower terms coincide. Here
$$
\omega(x)=x^3+c_2x^2+c_1x-c_1x^{-1}-c_2x^{-2}-x^{-3}.
$$
We can rewrite the right hand side of \eqref{ag4} as
$$
4\eta(\xi_1\xi_2\xi_3-1)\frac{\hat{a}_4(\xi_1,\xi_2,\xi_3,\eta)}{g_4(\xi_1,\xi_2,\xi_3,\eta)},
$$
where 
\begin{eqnarray*}
\hat{a}_4(\xi_1,\xi_2,\xi_3,\eta)&=&\xi_1^3\xi_2^3\xi_3^3\eta^3\left(a_4(\xi_1,\xi_2,\xi_3,\eta)-\tilde{a}_4(\xi_1,\xi_2,\xi_3,\eta)\right),\\
g_4(\xi_1,\xi_2,\xi_3,\eta)&=&\xi_1^3\xi_2^3\xi_3^3\eta^3 G^{\omega}(\xi_1,\xi_2,\xi_3,\eta) .
%=\xi_1^3\xi_2^3\xi_3^3\eta^3\left(\omega(\xi_1\xi_2\xi_3\eta)-\omega(\xi_1)-\omega(\xi_2)-\omega(\xi_3)-\omega(\eta)\right).
\end{eqnarray*}
Here both $\hat{a}_4(\xi_1,\xi_2,\xi_3,\eta)$ and $g_4(\xi_1,\xi_2,\xi_3,\eta)$ are polynomials in their arguments. One can show
%\marginpar{add?}
that the polynomial $g_4(\xi_1,\xi_2,\xi_3,\eta)$ is irreducible for any choice of parameters $c_1,c_2\in\k$. Moreover, the total degree of $g_4(\xi_1,\xi_2,\xi_3,\eta)$ is $24$, while the total degree of $\hat{a}_4(\xi_1,\xi_2,\xi_3,\eta)$ is less than $24$
since the equations are of form (\ref{gen}).
Therefore, the polynomial $g_4(\xi_1,\xi_2,\xi_3,\eta)$ cannot divide $\hat{a}_4(\xi_1,\xi_2,\xi_3,\eta)$. It follows from Theorem \ref{main} that
$\hu^3 (\phi_3(\xi_1,\xi_2,\xi_3,\eta)-\tilde{\phi}_3(\xi_1,\xi_2,\xi_3,\eta))$ is quasi-local, while the expression $4\hu^3\eta(\xi_1\xi_2\xi_3-1)\frac{\hat{a}_4(\xi_1,\xi_2,\xi_3,\eta)}{g_4(\xi_1,\xi_2,\xi_3,\eta)}$ is quasi-local if and only if $\hat{a}_4(\xi_1,\xi_2,\xi_3,\eta)=0$. Therefore, $a_4(\xi_1,\xi_2,\xi_3,\eta)=\tilde{a}_4(\xi_1,\xi_2,\xi_3,\eta)$, and thus $f_4=\tilde{f}_4$. 

By the same argument, we can show that if $f_i=\tilde{f}_i,\,\,i=1,\ldots,k$, then $f_{k+1}=\tilde{f}_{k+1}$, for $k\geq 4$.
\iffalse
Indeed,
\begin{eqnarray*}
&&\phi_k(\xi_1,\ldots,\xi_k,\eta)-\tilde{\phi}_k(\xi_1,\ldots,\xi_k,\eta)=(k+1)\frac{\eta(\xi_1\cdots\xi_k-1)(a_{k+1}(\xi_1,\ldots,\xi_k,\eta)-\tilde{a}_{k+1}(\xi_1,\ldots,\xi_k,\eta))}{\omega(\xi_1\cdots\xi_k\eta)-\omega(\xi_1)-\cdots-\omega(\xi_k)-\omega(\eta)}\\
&& \qquad =(k+1)\eta(\xi_1\cdots\xi_k-1)\frac{\hat{a}_{k+1}(\xi_1,\ldots,\xi_k,\eta)}{g_k(\xi_1,\ldots,\xi_k,\eta)},
\end{eqnarray*}
where 
\begin{eqnarray*}
\hat{a}_{k+1}(\xi_1,\ldots,\xi_k,\eta)&=&\xi_1^3\cdots\xi_k^3\eta^3\left(a_{k+1}(\xi_1,\ldots,\xi_k,\eta)-\tilde{a}_{k+1}(\xi_1,\ldots,\xi_k,\eta)\right),\\
g_k(\xi_1,\ldots,\xi_k,\eta)&=&\xi_1^3\cdots\xi_k^3\eta^3\left(\omega(\xi_1\cdots\xi_k\eta)-\omega(\xi_1)-\cdots-\omega(\xi_k)-\omega(\eta)\right).
\end{eqnarray*}
Again, the polynomial $g_k(\xi_1,\ldots,\xi_k,\eta)$ cannot divide $\hat{a}_{k+1}(\xi_1,\ldots,\xi_k,\eta)$, the element $\hu^k (\phi_k(\xi_1,\ldots,\xi_k,\eta)-\tilde{\phi}_k(\xi_1,\ldots,\xi_k,\eta))$ is quasi-local, while the element $4\hu^k\eta(\xi_1\cdots\xi_k-1)\frac{\hat{a}_{k+1}(\xi_1,\ldots,\xi_k,\eta)}{g_k(\xi_1,\ldots,\xi_k,\eta)}$ is quasi-local if and only if $\hat{a}_{k+1}(\xi_1,\ldots,\xi_4,\eta)=0$. Therefore, $a_{k+1}(\xi_1,\ldots,\xi_k,\eta)=\tilde{a}_{k+1}(\xi_1,\ldots,\xi_k,\eta)$, and thus $f_{k+1}=\tilde{f}_{k+1}$.
\fi
\end{proof}

{\bf Sketch of the proof of the theorem \ref{class3}}: The symbolic representation of a generic equation (or formal series)  (\ref{gen}) is of the form (\ref{eqsymb}) with
\[
\omega(\xi)=P_1(\xi)-P_1(\xi^{-1}),\quad 
a_s(\xi_1,\ldots,\xi_s)=P_s(\xi_1,\ldots,\xi_s)-P_s(\xi_1^{-1},\ldots,\xi_s^{-1}),\quad s=2,3,\ldots,
\]
where
\[
P_1(\xi)=\xi^3+\alpha\xi^2+\beta\xi,\quad\alpha,\beta\in\C,
\]
\[
 P_s(\xi_1,\ldots,\xi_s)=\left\langle\sum_{i_1=0}^3\sum_{i_2,\ldots,i_{s}=0}^2c_{i_1\ldots i_s}\xi_1^{i_1}\cdots\xi_s^{i_s}\right\rangle,\quad c_{i_1\cdots i_s}\in\C.
\]
The algorithm can be split into two steps.

\underline{Step 1}. The coefficients $\phi_1,\phi_2,\phi_3$ of the formal recursion operator with $\phi(\eta)=\eta$ can be found explicitly (Theorem  \ref{theorLam}). It is easy to show that terms $\hu\phi_1,\hu^2\phi_2$ are quasi-local for any $\alpha,\beta$ and constants $c_{i_1i_2},c_{i_1i_2i_3}$. The requirement of the quasi-locality of the term $\hu^3\phi_3(\xi_1,\xi_3,\xi_3,\eta)$ results in the system of polynomial equations on $\alpha,\beta$ and $c_{i_1i_2},c_{i_1i_2i_3},c_{i_1i_2i_3i_4}$ - the necessary integrability conditions. The action of the re-scaling group $u_k\to \mu u_k,\,k\in\Z,\,\,t\to \nu t,\,\,\mu,\nu\in\C^*$ on constants $\alpha,\beta$ and $c_{i_1i_2},c_{i_1i_2i_3},c_{i_1i_2i_3i_4}$ is of the form
\[
\alpha\to\alpha,\quad\beta\to\beta,\quad c_{i_1\cdots i_s}\to\mu^{s-1}\nu^{-1}c_{i_1\cdots i_s},
\]
and modulo the action of this group the set of solutions of the system of integrability conditions is finite. Thus we obtain the finite list of $3$-approximate integrable equations of the form (\ref{gen}). 

\underline{Step 2}. By Proposition \ref{uniq} the requirement of quasi-locality of $\hu^s\phi_s(\xi_1,\ldots,\xi_s,\eta),\,s\ge 4$ uniquely determines the terms  $\hu^sa_s(\xi_1,\ldots,\xi_s),\,\,s\ge 4$ and imposes further restrictions on the constants $\alpha,\beta,c_{i_1i_2},c_{i_1i_2i_3},c_{i_1i_2i_3i_4}$. The obtained sequence of coefficients $a_s(\xi_1,\ldots,\xi_s)$ either truncates at $s=5$ or continues indefinitely. In the former case this leads to polynomial equations of the above, and in the latter case, results in equations (\ref{ser3}), (\ref{ser2}) and (\ref{ser1}).

\underline{Step 3}. Integrability of every equation from the list is shown below either by transforming to a known integrable equation or by providing the Lax representation.
\hfill $\square$

The equations in Theorem \ref{class3} can be split into the following four lists:
\begin{itemize}
 \item[List 1:] Equations related to the stretched Volterra equation: \eqref{V1}--\eqref{ser3}. It is obvious that \eqref{V1}, 
 \eqref{V2} and \eqref{V3} are from the Volterra equation of the form $u_t=(\alpha+ \beta u+\gamma u^2)(u_1-u_{-1}), \alpha,\beta,\gamma\in\k.$ The other equations in this list are transformed into \eqref{V1}, for which we write as
 \begin{equation}\label{vol3}
 w_t=w (w_3-w_{-3})
 \end{equation}
 for clarity, by the following transformations:
 $$
 \eqref{eq1}: w=u u_1 u_2; \quad  \eqref{eq4}: w=u u_1; \quad \eqref{ser3}: w=\frac{u u_2}{u_1} .
 $$
 \item[List 2:] Linearisable equations: \eqref{ser2}, which is  related to the linear equation $w_t=w_3-w_{-3}$ 
 by the transformation $u=w w_1 w_2$.
 \item[List 3:] Equations related to the Narita-Itoh-Bogoyavlensky equation: \eqref{bg2}--\eqref{eq31}. Equation \eqref{bg2}
 is the well-known Narita-Itoh-Bogoyavlensky chain. All other are related to it. To be clear, we write it in different variable
 \begin{equation}\label{INB}
 w_{t}=w\left(w_3+w_2+w_1-w_{-1}-w_{-2}-w_{-3}\right) .
 \end{equation}
The transformations are as follows:
$$
\eqref{bg1}: w=u u_1 u_2; \ \eqref{bg11}: w=u u_1 u_2 u_3; \  \eqref{bg12}: w=u u_1 u_2 (u_3+1);
\  \eqref{eq51}:w=u u_2;\  \eqref{eq31}: w=u u_1.
$$
 \item[List 4:] Other equations: \eqref{eq2}--\eqref{ser1}. Equations \eqref{eq2} and \eqref{eq5} appeared in \cite{Adler2} as  discrete Sawada-Kotera equations and they are related by $u\to u u_1$. Equations \eqref{eq3} and \eqref{ser1} are new to best of our knowledge. We are going to show their integrability in Section \ref{sec63} by presenting their Lax representations.
 
 %For \eqref{ser1} we have found its Lax representation, which will be presented in Section \ref{sec62} where we are discussing integrable equations of higher orders.
\end{itemize}

\subsection{Integrable hierarchies of higher orders}\label{sec62}
The Narita-Itoh-Bogoyavlensky (NIB) lattice (\ref{NIB}) 
is a family of integrable equations of order  $(-n,n),\ n\in \N$. For any two distinct values of  $n$, the corresponding flows do not commute and belong to different integrable hierarchies. We shall call them $n$--{\sl relatives} of the NIB family.  Equations \eqref{NIB1}, \eqref{NIB2} and \eqref{bg2} are $1-,\ 2-$ and $3$--relatives respectively. In this section, we explore higher order relatives of all integrable equations listed in  Theorem \ref{class3}. 
%For each equation in Theorem \ref{class3}, we present one such integrable hierarchy of higher orders.
\begin{itemize}
 \item[List 1:] 
Equations \eqref{V1}, \eqref{V2} and \eqref{V3} are   particular cases   of the equation
$$
u_t=(\alpha+\beta u+\gamma u^2)(u_n-u_{-n}), \quad \alpha, \beta, \gamma\in \k
$$
which is obtained from the integrable Volterra type  equation $u_t=(\alpha+\beta u+\gamma u^2)(u_1-u_{-1})$   by stretching of the discrete variable.

The other equations in the List 1 can be transformed into \eqref{vol3} and their  integrable  higher order relatives can transformed into \begin{equation}\label{stren}
w_t=w(w_n-w_{-n}) .
\end{equation}
Equation (\ref{eq1}) is the $3$--relative of the family
\begin{eqnarray}
\nonumber
u_t=u\left(\prod_{k=1}^n u_k-\prod_{k=1}^n u_{-k}\right)-u^2\left(\prod_{k=1}^{n-1}u_k-\prod_{k=1}^{n-1}u_{-k}\right).
\end{eqnarray}
The latter can be mapped into \eqref{stren} by the transformation $w=\prod_{k=0}^{n-1}u_k$. 
\iffalse
Indeed, 
\begin{eqnarray*}
 &&w_t=\sum_{r=0}^{n-1} \frac{w}{u_r} \left(u_r\left(\prod_{k=1}^n u_{k+r}-\prod_{k=1}^n u_{r-k}\right)-u_r^2\left(\prod_{k=1}^{n-1}u_{k+r}-\prod_{k=1}^{n-1}u_{r-k}\right)\right)\\
 &&\quad=w \left(\prod_{k=1}^n u_{k+n-1}-\prod_{k=1}^n u_{-k}\right)=w(w_n-w_{-n}) .
\end{eqnarray*}
\fi
When $n=2$, it reduces to equation ${\rm (E.1^\prime)}$ in \cite{GGY19}.

The family
\begin{eqnarray*}
u_t=u\sum_{i=0}^{n-1} (-1)^i \left(u_{n-1-i} u_{n-i}-u_{i+1-n} u_{i-n}\right) ,
\end{eqnarray*}
is transformed to (\ref{stren})
by setting $w=uu_1$. It reduces to equation ${\rm (E.1^\prime)}$ in \cite{GGY19} when $n=2$ and to equation (\ref{eq4})
when $n=3$.

Equation \eqref{ser3} is the $1$--relative in the family 
\begin{eqnarray*}
%&&u_t=u \left(\frac{u_1 u_3 \cdots u_{2n-1}u_{2n+1}}{u_2 u_4 \cdots u_{2n-2}u_{2n}}
%+\frac{u u_2 \cdots u_{2n-2}u_{2n}}{u_1 u_3 \cdots u_{2n-3}u_{2n-1}}-
%\frac{u u_{-2} \cdots u_{2-2n}u_{-2n}}{u_{-1} u_{-3} \cdots u_{3-2n}u_{1-2n}} 
%-\frac{u_{-1} u_{-3} \cdots u_{1-2n}u_{-2n-1}}{u_{-2} u_{-4} \cdots u_{2-2n}u_{-2n}}\right)\\
&&u_t= u \left(\prod_{l=1}^n \frac{u_{2l-1}}{u_{2l}} u_{2n+1}+\prod_{l=0}^{n-1} \frac{u_{2l}}{u_{2l+1}} u_{2n}
-\prod_{l=0}^{n-1} \frac{u_{-2l}}{u_{-1-2l}} u_{-2n}-\prod_{l=1}^n \frac{u_{1-2l}}{u_{-2l}} u_{-2n-1}\right),
\end{eqnarray*}
which reduces  to 
$
w_t=w (w_{2n+1}-w_{-2n-1})
$
by the transformation $w=\prod_{l=0}^{n-1} \frac{u_{2l}}{u_{2l+1}} u_{2n}$. 
\item[List 2:]
Equation (\ref{ser2}) can be linearised. In general, we let $u=\prod_{k=0}^{n-1}w_k$, where $w$ satisfies a linear equation 
$w_t=w_n-w_{-n}$. 
\iffalse
\begin{eqnarray*}
u_t=u(\sum_{k=1}^{n}\frac{u_{k}}{u_{k-1}}-\sum_{k=1}^{n}\frac{u_{-k}}{u_{1-k}}),
\end{eqnarray*}
which can be linearised, that is, it is obtained from $w_t=w_n-w_{-n}$ by letting $u=\prod_{k=0}^{n-1}w_k$. 
\fi
Note that 
$\frac{u_{k+1}}{u_{k}}=\frac{w_{n+k}}{w_{k}} $ and
$\frac{u_{-k}}{u_{1-k}}=\frac{w_{-k}}{w_{n-k}} $ for $k>0$. Thus
\begin{eqnarray}\label{linfam}
 u_t=\sum_{k=0}^{n-1} \frac{u}{w_k} (w_{n+k}-w_{k-n})%=u \sum_{k=0}^{n-1}  (\frac{w_{n+k}}{w_k}-\frac{w_{k-n}}{w_k})
 =u \sum_{k=0}^{n-1}  (\frac{u_{k+1}}{u_k}-\frac{w_{-1-k}}{w_{n-1-k}})
  =u \sum_{k=0}^{n-1}  (\frac{u_{k+1}}{u_k}-\frac{u_{-1-k}}{u_{-k}}) .
\end{eqnarray}
Thus, $1$--relative of the family (\ref{linfam}) is a linear equation, $2$--relative is a special case of  ${\rm (E.8)}$ in \cite{GGY19}  and $3$--relative is (\ref{ser2}).
\item[List 3:]
As we mentioned at the beginning of this section,  equation \eqref{bg2} is $3$--relative of the NIB family (\ref{NIB}), which here we  write in different variable for convenience
 \begin{equation}\label{INBn}
 w_{t}=w\left(w_n+\cdots+w_1-w_{-1}-\cdots-w_{-n}\right) .
 \end{equation}
The families corresponding to equations  \eqref{bg1}--\eqref{eq31} we denote as $\eqref{bg1}_n-\eqref{eq31}_n$. They are can be transformed  to (\ref{INBn}) by polynomial maps:
\begin{eqnarray}
&& \eqref{bg1}_n:\quad u_{t}=u\left(\prod_{j=1}^n u_j-\prod_{j=1}^n u_{-j}\right), \quad \quad \quad w=\prod_{j=0}^{n-1} u_j;\label{bg1n} \\
&&\eqref{bg11}_n:\quad u_{t}=u^2\left(\prod_{j=1}^n u_j-\prod_{j=1}^n u_{-j}\right), \quad\quad \ \ w=\prod_{j=0}^{n} u_j;\label{bg11n} \\
&&\eqref{bg12}_n:\quad u_{t}=(u^2+ u)\!\!\left(\prod_{j=1}^n u_j-\prod_{j=1}^n u_{-j} \!\!\right), \quad
w=\prod_{j=0}^{n-1} u_j (u_n+1);\label{bg12n}\\  
&& \eqref{eq51}_n:\quad u_{t}=u \sum_{i=0}^{\lfloor{\frac{n-1}{2}}\rfloor}\left(u_{\lfloor{\frac{n}{2}}\rfloor-i} u_{n-i}-u_{i-\lfloor{\frac{n}{2}}\rfloor} u_{i-n}\right),
\quad \quad w=u u_{\lceil{\frac{n}{2}}\rceil};\label{eq51n}\\  
&&\eqref{eq31}_n:\quad u_{t}=u \sum_{l=0}^{\lfloor{\frac{n-1}{2}}\rfloor}(u_{n-2l-1} u_{n-2l} -u_{2l+1-n} u_{2l-n}),\quad\ w=u u_1. \label{eq31n}
\end{eqnarray}
Here $\lfloor \cdot \rfloor $ and $\lceil \cdot \rceil $ are floor and ceiling functions respectively, and these can be checked by direct computation. For example, for (\ref{bg12n}),  we have
\begin{eqnarray*}
 w_t&=&w \sum_{k=0}^{n-1} (u_k+1) \left(\prod_{j=1}^n u_{j+k}-\prod_{j=1}^n u_{k-j}\right)+w u_n \left(\prod_{j=1}^n u_{n+j}-\prod_{j=1}^n u_{n-j}\right)\\
 &=&w \sum_{k=0}^{n-1} \left(w_k -\prod_{j=0}^{n-1} u_{j+k}+\prod_{j=0}^{n-1} u_{j+1+k}-w_{k-n}\right)
 +w \left(w_n -\prod_{j=1}^n u_{n+j}-w+\prod_{j=0}^{n-1} u_{j}\right),
\end{eqnarray*}
which is (\ref{INBn}) after straightforward simplification. 

The families \eqref{bg1} and \eqref{bg11} are two known modifications of the Narita-Itoh-Bogoyavlensky family \cite{bogo}. The 
family (\ref{bg12n}) can be found in a recent paper \cite{Adler3} on discrete integrable equations of higher order.

\item[List 4:]
Equation (\ref{eq2}) and its family 
\begin{eqnarray}
u_t=u^2\left(\prod_{k=1}^n u_k-\prod_{k=1}^n u_{-k}\right)-u\left(\prod_{k=1}^{n-1}u_k-\prod_{k=1}^{n-1}u_{-k}\right).\label{adler}
\end{eqnarray}
have been previously found by Adler and Postnikov \cite{Adler2}, where the authors also presented the corresponding fractional Lax representation.
This family can be viewed as an inhomogeneous deformation of the modified Narita-Itoh-Bogoyavlensky family (\ref{bg11n}). When $n=2l+1, l\geq 1$ we can use the transformation $w=u u_1$
\iffalse. Then equation (\ref{adler}) becomes
$$
u_t=u\left(\prod_{k=0}^l w_{2k}-\prod_{k=0}^l w_{-2k-1}\right)-u\left(\prod_{k=0}^{l-1}w_{2k+1}-\prod_{k=1}^{l}w_{-2k}\right).
$$
Thus it follows that
\fi
which maps the family (\ref{adler}) into the family
\begin{eqnarray}\label{eq5n}
 w_t=w\sum_{i=0}^1\left(\prod_{k=0}^l w_{2k+i}-\prod_{k=0}^l w_{-2k-i}-\prod_{k=1}^{l}w_{2k-i}+\prod_{k=1}^{l}w_{-2k+i}\right).
\end{eqnarray}
Equation (\ref{eq5}) is $1$--relative of the family (\ref{eq5n}).

It seems that equation (\ref{eq3}) has not appeared in the literature previously.
%for which we haven't found its Lax representation although we are able to find its higher order symmetries for fixed orders. 
It belongs to a new family of integrable equations\footnote{After we have completed this work, we found out that this equation and the family (\ref{neweq1}) were known to V.E. Adler.
He kindly sent us his unpublished notes, also containing their Lax representations. Under his permission, we present the Lax representation in Section \ref{sec7} for completeness.} 
\begin{eqnarray}
u_t
&=&u (\cS^{2n-1}-1) \left((u-1)\prod _{l=1}^{n-1} u_{-l}+(u_{1-2n}-1) \prod_{l=n}^{2n-2} u_{-l} 
\right)\label{neweq1} .
\end{eqnarray}
\iffalse
\begin{eqnarray}
u_t&=&u\left(\prod_{l=n}^{2n-1} u_l+\prod _{l=0}^{n-1} u_l-\prod _{l=0}^{n-1} u_{-l}-\prod_{l=n}^{2n-1} u_{-l}\right)
-u\left(\prod_{l=n}^{2n-2} u_l+\prod _{l=1}^{n-1} u_l-\prod _{l=1}^{n-1} u_{-l}-\prod_{l=n}^{2n-2} u_{-l}\right)\nonumber\\
&=&u (\cS^{2n-1}-1) \left((u-1)\prod _{l=1}^{n-1} u_{-l}+(u_{1-2n}-1) \prod_{l=n}^{2n-2} u_{-l} 
\right)\label{neweq1} .
\end{eqnarray}\fi
Its $1$--relative  reduces to the Volterra chain (\ref{NIB1}), while $2$--relative coincides with (\ref{eq3}). 
We also explicitly check that the $3-$ and $4$--relatives of the family, namely
\begin{eqnarray*}
&&u_t=u(\cS^5-1)\left( (u-1) u_{-1} u_{-2}+u_{-3} u_{-4} (u_{-5}-1)\right),\\
&&u_t=u(\cS^7-1)\left( (u-1) u_{-1} u_{-2} u_{-3}+ u_{-4} u_{-5}u_{-6} (u_{-7} -1)\right),
\end{eqnarray*}
possess quasi-local canonical formal recursion operators. 

\iffalse
It is obvious that $\ln u$ is a conserved density for equation \eqref{neweq1}. Moreover, it has a following $n-1$-th order conserved density.
\begin{Pro}
For any $n\in \N$, $(u_{n-1}-1) \prod _{l=0}^{n-2} u_l$ is a conserved density for \eqref{neweq1}.
\end{Pro}
\begin{proof} Let $w=(u_{n-1}-1) \prod _{l=0}^{n-2} u_l$. Then \eqref{neweq1} becomes
$$
u_t=u (\cS^{2n-1}-1) \left(w_{1-n}+w_{1-2n}+\prod_{l=n+1}^{2n-1} u_{-l} - \prod_{l=n}^{2n-2} u_{-l} \right).
$$
We now compute 
\begin{eqnarray*}
w_t&=&w \sum_{l=0}^{n-1} (\cS^{2n-1}-1) \left(w_{1+l-n}+w_{1+l-2n}+\prod_{j=n+1}^{2n-1} u_{l-j} - \prod_{j=n}^{2n-2} u_{l-j} \right)\\
&&+\left(\prod _{l=0}^{n-2} u_l \right) (\cS^{2n-1}-1) \left(w+w_{-n}+\prod_{l=2}^{n} u_{-l} - \prod_{l=1}^{n-1} u_{-l} \right)\\
%&\equiv&w (w_{2n-1}+\cdots w_1-w_{-1}-\cdots -w_{1-2n})-2 w (\cS^{2n-1}-1) \prod_{j=1}^{n-1} u_{-j}\\
%&&+\left(\prod _{l=0}^{n-2} u_l \right) (\cS^{2n-1}-1) \left(\prod_{l=2}^{n} u_{-l} - \prod_{l=1}^{n-1} u_{-l} \right)\\
&\equiv&w (w_{2n-1}+\cdots w_1-w_{-1}-\cdots -w_{1-2n})-2 \left(\prod _{l=0}^{n-1} u_l \right) (\cS^{2n-1}-1) \prod_{j=1}^{n-1} u_{-j}\\
&&+\left(\prod _{l=0}^{n-2} u_l \right) (\cS^{2n-1}-1) \left(\prod_{l=2}^{n} u_{-l} + \prod_{l=1}^{n-1} u_{-l} \right)\\
&=& \sum_{l=1}^{2n-1} (\cS^l-1) w w_{-l} -2 (\cS^{n-1}-1) \prod_{j=1-n}^{n-1}u_j
+(\cS^n-1) \!\!\!\prod_{-1\neq j=-n}^{n-2}\!\!\!\! u_j+(\cS^{n-1}-1)\prod_{j=1-n}^{n-2}u_j,
\end{eqnarray*}
where $\equiv$ denotes the equivalence under the image of $\cS-1$. Thus $w$ is a conserved density as stated.
\end{proof}
\fi

Equation (\ref{ser1}) is another new equation. Its family and Lax representation we discuss in the next Section \ref{sec63}.
\end{itemize}
\begin{Rem}
We have already seen that one low order equation may belong to several integrable families. For example the Volterra equation (\ref{NIB1}) is $1$--relative in (\ref{NIB}) and (\ref{neweq1}) families. Equation \eqref{eq31} belongs to the family (\ref{eq31n}) as well as the family 
$$
u_{t}=u \left((1+\cS^{-\lceil\frac{n}{2}\rceil})\prod_{i=0}^{\lfloor{\frac{n-1}{2}}\rfloor} u_{n-i}
-(1+\cS^{\lceil\frac{n}{2}\rceil})\prod_{i=0}^{\lfloor{\frac{n-1}{2}}\rfloor} u_{i-n}\right),\quad w=\prod_{i=0}^{\lfloor{\frac{n-1}{2}}\rfloor}u_i.
$$
which has the same $w$ ancestor (\ref{INBn}). 
\end{Rem}

\subsection{Integrable deformation of the Narita-Itoh-Bogoyavlensky lattice}\label{sec63}
In this section, we study the integrability of the equation
\begin{eqnarray}\label{squareq}
u_t=(1+u^2)(u_n\prod_{k=1}^{n-1}\sqrt{1+u_k^2}-u_{-n}\prod_{k=1}^{n-1}\sqrt{1+u_{-k}^2}), \quad  n\in \mathbb{N}.
\end{eqnarray}
It can be viewed as an inhomogeneous deformation of the Narita-Itoh-Bogoyavlensky equation (\ref{bg11n}). Indeed, after the re-scaling $u_k\to \varepsilon^{-1} u_k,\ t\to \varepsilon^{n+1}t$ and in the limit $\varepsilon\to 0$ it turns into (\ref{bg11n}).
In the case $n=2$ its integrability was established in \cite{GY17}. 

\begin{Thm}\label{sqrt}
Equation \eqref{squareq}
possesses a Lax representation  $L_t=[A, L]$ with a rational (pseudo-difference) operator $L=Q^{-1}P$, where 
\begin{eqnarray*}
 Q=u-w u_1 \cS^{-1}, \quad P=(u w_1 \cS-u_1) \cS^{n-1},\qquad  w=\sqrt{1+u^2},
\end{eqnarray*}
and a skew-symmetric difference operator $A=L_{+}-(L_{+})^{\dagger}$. Here $L_+$ denotes the polynomial part of the Laurent representation 
for the rational operator $L$. 
\end{Thm}
{\bf Proof}. Note that
$$
Q^{-1}=(1-\frac{w u_1}{u} \cS^{-1})^{-1}\frac{1}{u}=-(1+\frac{w u_1}{u} \cS^{-1}+\frac{w u_1}{u} \cS^{-1}\frac{w u_1}{u} \cS^{-1}+\cdots)\frac{1}{u}.
$$
Thus 
\begin{eqnarray*}
&&L_{+}=a^{(n)}\cS^n+a^{(n-1)}\cS^{n-1} \cdots +a^{(1)}\cS,\\
&&a^{(n)}=w_1, \ a^{(n-1)}=u u_1, \ a^{(n-l)}=w w_{-1} \cdots w_{2-l} u_1 u_{1-l},\ 2\leq l\leq n-1.
\end{eqnarray*}
These are the coefficients of positive shifts in $A$. For the negative parts, we have
$$ a^{(-l)}=-\cS^{-l} a^{(l)}=-a^{(l)}_{-l}. $$

The Lax equation $L_t=[A, L]$ implies that there exists a difference operator $B$ such that
\begin{eqnarray}\label{lpq}
 P_t=BP-PA, \qquad Q_t=BQ-QA,
\end{eqnarray}
where the operator $B=\sum_{l=-n}^n b^{(l)}\cS^l$
can be determined by $A$, $P$ and $Q$. To do so, we write down the equivalent equations for (\ref{lpq}) according to the orders of $\cS$. Using the equation for the operator $Q$, we get
\begin{eqnarray}
 \cS^{n}:&& b^{(n)} u_n =u a^{(n)};\label{q1}\\
 \cS^{l}: && b^{(l+1)}u_{l+2}w_{l+1}-b^{(l)} u_{l}-u_1 w a_{-1}^{(l+1)}+u a^{(l)}=0, \ 1\leq l\leq n-1\ \mbox{or} \ -n\leq l\leq -2;\label{q2}\\
 \cS^0:&& u_{t}=-b^{(1)} u_{2} w_1+b^{(0)} u+u_1 w a_{-1}^{(1)};\label{q3}\\
 \cS^{-1}: && u_{1t} w+u_1 w_{t}= b^{(0)}u_1 w-b^{(-1)} u_{-1}+u a^{(-1)};\label{q4}\\
 \cS^{-n-1}: && b^{(-n)} u_{1-n}w_{-n}=u_1 w  a_{-1}^{(-n)}=-u_1 w w_{-n}.\label{q5}
\end{eqnarray}
From (\ref{q1}), it follows that $b^{(n)}=\frac{u w_1}{u_n}$. We substitute it into (\ref{q2}) for $l=n-1$ and obtain
\begin{eqnarray*}
 &&b^{(n-1)}= \frac{1}{u_{n-1}} \left(b^{(n)}u_{n+1}w_{n}-u_1 w a_{-1}^{(n)}+u a^{(n-1)}\right)\\
 &&\quad=\frac{1}{u_{n-1}} \left(\frac{u w_1}{u_n}u_{n+1}w_{n}-u_1 w2+u^2 u_1\right)=\frac{u u_{n+1} w_1w_n}{u_n u_{n-1}}
 -\frac{u_1}{u_{n-1}}.
\end{eqnarray*}

When $1\leq l\leq n-2$ notice that
$
-u_1 w a_{-1}^{(l+1)}+u a^{(l)}=0. %\quad 1\leq l\leq n-2.
$
Thus in this case (\ref{q2}) becomes $$b^{(l)} u_{l}=b^{(l+1)}u_{l+2}w_{l+1},$$
which leads to
\begin{eqnarray*}
b^{(l)}=\frac{u_{n-1}u_n}{u_l u_{l+1}}w_{l+1} \cdots w_{n-1}  b^{(n-1)}=
\frac{1}{u_l u_{l+1}} \left(\prod_{j=l+1}^{n-1} w_{j}\right) \left(u u_{n+1} w_1w_n-u_1 u_n\right).
\end{eqnarray*}
Similarly, from (\ref{q5}), we get $b^{(-n)}=-\frac{u_1 w}{ u_{1-n}}$. Substituting it into (\ref{q2}) for $l=-n$, we have
\begin{eqnarray}
 &&b^{(-n+1)}= \frac{1}{u_{2-n}w_{1-n}} \left(b^{(-n)}u_{-n}+u_1 w a_{-1}^{(1-n)}-u a^{(-n)}\right)\nonumber\\
 &&\quad=\frac{1}{u_{2-n}w_{1-n}} \left(-\frac{u_1 w}{ u_{1-n}} u_{-n}-u_1 w u_{-n} u_{1-n}+u w_{1-n}\right)
=\frac{u u_{1-n}-u_1 u_{-n} w w_{1-n}}{u_{1-n} u_{2-n}} \ .\label{bmn}
\end{eqnarray}

When $1-n\leq l\leq -2$, notice that
\begin{eqnarray}
 &&-u_1 w a_{-1}^{(l+1)}+u a^{(l)}=u_1 w \cS^{l} a^{(-l-1)}-u \cS^l a^{(-l)}
=u_1 w \left(\prod_{k=0}^{n+l-1}w_{l-k} \right) u_{l+1} u_{-n}-u\left(\prod_{k=0}^{n+l-2}w_{l-k} \right) u_{l+1} u_{1-n}\nonumber\\
 &&\qquad=-\left(\prod_{k=0}^{n+l-2}w_{l-k} \right) u_{l+1} (u u_{1-n}-u_1 u_{-n} w w_{1-n}) .\label{bmin}
\end{eqnarray}
Thus combining (\ref{bmn}) and (\ref{bmin}), using (\ref{q2}) we obtain, for $2-n\leq l\leq -1$,
\begin{eqnarray*}
b^{(l)}=
\frac{1}{u_l u_{l+1}} \left(\prod_{j=2-n}^{l} w_{j}\right) \left(u u_{1-n}-u_1 u_{-n} w w_{1-n}\right).
\end{eqnarray*}
We now eliminate $b^{(0)}$ from (\ref{q3}) and (\ref{q4})  and get
\begin{eqnarray*}
%b^{(0)}=\frac{1}{u} \left(u_{t}+b^{(1)} u_{2} w_1-u_1 w a_{-1}^{(1)} \right),
&&\qquad u_{1t} w+u_1 w_{t}-\frac{u_1 w}{u} u_t=w u_{1t}-\frac{u_1}{uw} u_t\\
&&=\frac{u_1 w}{u} \left(b^{(1)} u_{2} w_1-u_1 w a_{-1}^{(1)} \right)-b^{(-1)} u_{-1}+u a^{(-1)}
=\frac{u_1 u_2 w w_1}{u} b^{(1)} -\frac{u_1^2 w^2}{u} a_{-1}^{(1)} -b^{(-1)} u_{-1}-u a_{-1}^{(1)}\\
&&=\frac{w w_1}{u} \left(\prod_{j=2}^{n-1} w_{j}\right) \left(u u_{n+1} w_1w_n-u_1 u_n\right)
-\frac{1}{u} \left(\prod_{j=1}^{n-2} w_{-j}\right) \left(u u_{1-n}-u_1 u_{-n} w w_{1-n}\right)\\
&&\quad -\frac{u_1^2 w^2+u^2}{u} \left(\prod_{j=1}^{n-2} w_{-j}\right) u u_{1-n}\\
&&=\left(w \cS-\frac{u_1}{uw} \right) w^2 \left( \left(\prod_{j=1}^{n-1} w_{j}\right) u_n-\left(\prod_{j=1}^{n-1} w_{-j}\right) u_{-n}\right),
\end{eqnarray*}
which is satisfied by the given equation. 
Finally, we check the obtained $B$ is consistent with the equation for the operator $P$ in (\ref{lpq}) by directly computation,
and thus we proved the statement. \hfill $ \Box$

\iffalse
\begin{eqnarray}
 \cS^{n+1}:&& b^{(n)} u_n w_{n+1}=u w_1 a_n^{(n)};\label{p1}\\
 \cS^{l}: && b^{(l-1)}u_{l-1}w_l-b^{(l)} u_{l+1}-u w_1 a_n^{(l-1)}+u_1 a_{n-1}^{(l)}=0, \ 2\leq l\leq n\ \mbox{or} \ 1-n\leq l\leq -1;\label{p2}\\
 \cS:&& u_t w_1+u w_{1t}= b^{(0)}u w_1-b^{(1)} u_{2}+u_1 a_{n-1}^{(1)};\label{p3}\\
 \cS^0: && u_{1t}=-b^{(-1)} u_{-1} w+b^{(0)} u_1+u w_1 a_n^{(-1)};\label{p4}\\
 \cS^{-n}: && b^{(-n)} u_{1-n}=u_1 a_{n-1}^{(-n)}=-u_1 w.\label{p5}
\end{eqnarray}
\fi
\begin{Rem}
In \cite{GY17}, Garifullin and Yamilov established the integrability properties of the equation
\begin{eqnarray*}\label{gy17}
u_t=(u^2-1)\left(\sqrt{u_1^2-1} u_2-\sqrt{u_{-1}^2-1} u_{-2} \right),
\end{eqnarray*}
which is the equation (\ref{squareq}) when $n=2$ under the scaling transformation.
They provided its Lax pair as $L \psi=0$ and $\psi_t= A \psi$ with
\begin{eqnarray*}
&& L=u w'_1 \cS^2+u_1 \cS -\lambda \left(u_1 w' \cS^{-1} -u \right);\qquad w'=\sqrt{u^2-1},\\
&& A= \frac{w'}{u} \left(w'(u_1 \cS+u_{-1} \cS^{-1})-\lambda^{-1} u_{-1} \cS+\lambda u_1 \cS^{-1}\right).
\end{eqnarray*}
This is the same as the Lax representation given in the theorem (after rescaling and simple gauge).
Let $P=u w'_1 \cS^2+u_1 \cS$ and $Q=u_1 w' \cS^{-1} -u$. 
If we eliminate $\lambda$ in $A$, we get
\begin{eqnarray*}
 &&A=\frac{w'^2 u_1}{u}  \cS+\frac{w'^2 u_{-1}}{u} \cS^{-1}- \frac{w' u_{-1}}{u} \cS P^{-1} Q
 +\frac{w' u_1}{u}\cS^{-1} Q^{-1} P\\
&&=\frac{w'^2 u_1}{u}  \cS+\frac{w'^2 u_{-1}}{u} \cS^{-1}- \frac{1}{u} \cS^{-1}  (u_1 w' \cS^{-1} -u) +P^{-1}Q
 +\frac{1}{u} (u w'_1 \cS^2+u_1 \cS)+Q^{-1}P\\
 &&=w'_1 \cS^2 + u u_1 \cS+u u_{-1} \cS^{-1} -w'_{-1} \cS^{-2} +P^{-1}Q+Q^{-1}P,
\end{eqnarray*}
where the nonlocal part commutes with the fractional operator $Q^{-1}P$ and thus we can discard it. 
\end{Rem}
We prove that the hierarchy of commuting symmetries is well defined for any $n\in \mathbb{N}$. To do so, we use the $r$-matrix approach \cite{BM94}. We consider the Lie algebra denoted by $\lieg$ of the formal Laurent series of the shift operator with
the commutator $[A, B]=AB-BA$, where $A, B\in \lieg$.  It is easy to see that any element
$$
B= b^{(m)} \cS^m + b^{(m-1)} \cS^{m-1} + \cdots \in \lieg
$$
admits a unique decomposition of the form
$$
B=B_{+}-(B_{+})^{\dagger}+H,
$$
where $B_{+}$ is a difference operator with positive powers of $\cS$. We denote the antisymmetric part $B_{+}-(B_{+})^{\dagger}$ 
of any element $B$ as $\lieg_{+}$ and the other part $H$ as $\lieg_{-}$. It is obvious that these both parts form Lie subalgebras.
Thus we have the decomposition of the Lie algebra
$$
\lieg= \lieg_{+} \oplus \lieg_{-}
$$
into the direct sum of two Lie subalgebras. We define the projections 
$$
\pi_{\pm}: \lieg \rightarrow \lieg_{\pm}
$$
and the $r$-matrix $r=\frac{1}{2} (\pi_{+}-\pi_{-}).$ We now formulate the statement generating commuting symmetries, which can be proved in the standard way in the $r$-matrix approach \cite{BM94}.

{\bf Corollary}. The flows defined by the Lax equations
$
\partial_{t_p} L =[\pi_{+}(L^p), \ L]
$
commute with each other. % that is, $\partial_{t_p}\partial_{t_q} L= \partial_{t_q}\partial_{t_p} L$.

{\bf Proof}.  Let $A^{(p)}=\pi_{+}(L^p)$. Using Lax equations, we have
\iffalse
\begin{eqnarray*}
\partial_{ t_q} \partial_{t_p} L
=[\partial_{ t_q}A^{(p)}, L]+
[A^{(p)}, \partial_{ t_q}L]
=[\partial_{ t_q}A^{(p)}, L]+
[A^{(p)}, [A^{(q)}, L]]
\end{eqnarray*}
Hence
\fi
\begin{eqnarray*}
\partial_{ t_q} \partial_{t_p} L-\partial_{ t_p} \partial_{t_q} L
=[\partial_{ t_q}A^{(p)}
-\partial_{ t_p}A^{(q)}-[A^{(q)}, A^{(p)}],L].
\end{eqnarray*}
So it is sufficient to show that
$$
\partial_{ t_q}A^{(p)}
-\partial_{ t_p}A^{(q)}-[A^{(q)}, A^{(p)}]=0.
$$
For $l\in \mathbb{Z}$, we have $\partial_{t_p} L^l=[A^{(p)},L^l] $ since
\begin{eqnarray*}
[\partial_{t_p} L^l, \ L]=-[L^l, \partial_{t_p}L]
=-[L^l, [A^{(p)}, L]]
=[[A^{(p)},L^l], L]\ .
\end{eqnarray*}
Thus
\begin{eqnarray*}
&&\partial_{ t_q}A^{(p)}
-\partial_{ t_p}A^{(q)}-[A^{(q)}, A^{(p)}]
\\&&
=\pi_{+}(\partial_{ t_q}L^p)
-\pi_{+}(\partial_{ t_p}L^q)-[A^{(q)}, A^{(p)}]
=\pi_{+}([A^{(q)}, L^p]) -\pi_{+}([A^{(p)}, L^q])
-[A^{(q)}, A^{(p)}]
\\&&
=\pi_{+}([-\pi_{-}(L^q), L^p]) -\pi_{+}([-\pi_{-}(L^p), L^q])
-[L^q-\pi_{-}(L^q), L^p-\pi_{-}(L^p)]
\\&&
=\pi_{+}(-[\pi_{-}(L^q), \pi_{-}(L^p)])=0.
\end{eqnarray*}
This leads to
$\partial_{ t_q} \partial_{t_p} L=
\partial_{ t_p} \partial_{t_q} L$ as required.
\hfill $\diamond$

\section{Summary and Discussion}\label{sec7}

In order to give an exhaustive description of all integrable differential-difference equations of certain type one needs to find strong and verifiable necessary integrability conditions. All previous attempts to tackle this problem were   based on the integrability conditions dependent of the symmetry structure of equations which may not be known in advance.
\iffalse
To solve the classification problem of differential-difference hierarchies it is critical to derive easily verifiable necessary integrability conditions. All known classification results are obtained so far based on the integrability conditions dependent of the symmetry structure of such equations. 
\fi
In this paper, we developed an approach to establish the universal integrability conditions by introducing the notion of quasi-locality in the context of symbolic representation. We proved that if an equation of the form
\begin{equation*}
 u_t=f=f(u_{-n},\ldots, u_n),\qquad f=\sum_{i\ge 1}f^{(i)},\qquad f^{(i)}\in\ring^i,\qquad \ \ f^{(1)}\ne 0,
\end{equation*}
possesses an infinite dimensional algebra of its symmetries,  then there exists a unique quasi-local formal recursion operator with
symbolic representation
$$
\eta+\sum_{p\ge 1}\hu^p\phi_p(\xi_1,\ldots,\xi_p,\eta).
$$
Moreover, the recursive formulae for all terms $\phi_p(\xi_1,\ldots,\xi_p,\eta)$ are explicitly given. This is achieved by developing symbolic representation for the difference polynomial ring, difference operators and formal series.
The requirement of quasi-locality for all terms $\phi_p(\xi_1,\ldots,\xi_p,\eta)$ leads to the necessary integrability conditions, which are independent on the structure of the symmetry algebra.

We applied our new approach to classification of integrable equations of the form (\ref{gen0}).
We reproduced the integrable equations of such type in \cite{Yami1} for $n=2$, and obtained a complete list of integrable equations for $n=3$, which includes 17 equations which are complimentary to the lower order hierarchies.
%without lower order symmetries.
%including two new equations (\ref{ser1}) and \eqref{eq3}. 
For each equation in the list we found an infinite family of integrable equations of arbitrary high order. Integrability of the families obtained we proved in each case either by providing a transformations to a known integrable family, or by presenting a Lax representation.

Here, for completeness, we would like present a fractional Lax representation 
\[ L_t=[A,L],\ L=Q^{-1}P,\ \Leftrightarrow\ P_t=BP-PA, \quad Q_t=BQ-QA 
\]
 for the new family of equations \eqref{neweq1} found by Adler \cite{Adler4} in which 
\begin{eqnarray*}
 &&P=\left( \prod_{l=0}^{n-1} u_l \right)\cS^{2n}+\cS; \ \ 
 A= -\left(\prod_{l=-n}^{n-2} u_l\right) \cS^{2n-1} +\sum_{j=0}^{2n-2} \left( (1-u_{j-n})
 \prod_{l=n+1}^{2n-1} u_{j-l} \right)+\cS^{1-2n}; \\
 &&Q=\left( \prod_{l=0}^{n-2} u_l \right)\cS^{2n-1}+1; \ \ 
 B=-\left(\prod_{l=0}^{2n-2} u_l \right) \cS^{2n-1} +\sum_{j=0}^{2n-2} \left( (1-u_{j-2n+1})
 \prod_{l=n}^{2n-2} u_{j-l} \right)+\cS^{1-2n}.
\end{eqnarray*}

The global classification problem of integrable differential-difference equations still remains a very challenging problem. In the case of scalar polynomial 
evolutionary equations, we believe that there are new integrable hierarchies of high orders starting from any order.  Thus it is impossible to achieve the global classification result as it was done for scalar evolution PDEs \cite{mr99g:35058, mr2001h:37147}. The higher the order $n$ is in the equation, the more integrable equations there are even after we remove the ones being symmetries of lower order ones. Thus it is important to develop an approach to identify the transformations between integrable equations and to develop the concept of integrable families which have been discussed, but not yet defined in this paper.

Integrable differential-difference equations are generalised symmetries for discrete integrable equations. The authors of \cite{GGY19} constructed the autonomous quad-equations which admit as symmetries five-point differential-difference equations belonging to classification lists in \cite{Yami1}. Recently it was shown that there exist multi-points integrable discrete equations on a plain lattice related to B{\"a}cklund-Darboux transformations for the Narita-Itoh-Bogoyavlensky equation (\ref{NIB}) \cite{Adler3, pavlos19}.
 It would be interesting to construct discrete integrable equations for families discussed in Section \ref{sec62},
in particular, for the new integrable families \eqref{neweq1} and \eqref{squareq}.

\section*{Acknowledgments}
All authors would like to thank V.E. Adler for useful discussions and sending us his unpublished research notes.
The paper is supported by EPSRC small grant scheme EP/V050451/1, and partially by 
grants EP/P012655/1 and EP/P012698/1. AVM is grateful to the Ministry of Sciences  and Higher Education RF 
(agreement   075-02-2021-1397) for partial support. 

\end{document}